\documentclass[twocolumn]{article}

\usepackage{enumerate}
\usepackage{amsmath}
\usepackage{amsthm}
\usepackage{amssymb}
\usepackage{subfig}
\usepackage{algorithm}
\usepackage{algorithmic}
\usepackage{url}
\usepackage{graphicx}
\usepackage{times}
\usepackage{multirow}
\usepackage{epstopdf}
\usepackage{balance}
\usepackage{xspace}
\usepackage{xcolor}
\usepackage{fullpage}

\DeclareGraphicsExtensions{.pdf,.jpg,.png,.eps}

\newcommand{\comment}[1]{}
\newcommand{\figwidth}{0.99\columnwidth}

\newtheorem{theorem}{Theorem}

\newtheorem{definition}{Definition}

\newtheorem{observation}{Observation}

\title{Minimally Infrequent Itemset Mining using Pattern-Growth Paradigm and
Residual Trees}

\author{
\begin{tabular}{ccc}
	Ashish Gupta & Akshay Mittal & Arnab Bhattacharya \\
	\url{ashgupta@cse.iitk.ac.in} & \url{amittal@cse.iitk.ac.in} & \url{arnabb@iitk.ac.in} \\
	\multicolumn{3}{c}{Dept. of Computer Science and Engineering} \\
	\multicolumn{3}{c}{Indian Institute of Technology, Kanpur} \\
	\multicolumn{3}{c}{Kanpur, India} \\
\end{tabular}
}

\date{}

\begin{document}

\maketitle

\begin{abstract}
	Itemset mining has been an active area of research due to its successful
	application in various data mining scenarios including finding association
	rules.  Though most of the past work has been on finding frequent itemsets,
	infrequent itemset mining has demonstrated its utility in web mining,
	bioinformatics and other fields.  In this paper, we propose a new algorithm
	based on the pattern-growth paradigm to find minimally infrequent itemsets.
	A minimally infrequent itemset has no subset which is also infrequent.  We
	also introduce the novel concept of residual trees.  We further utilize the
	residual trees to mine multiple level minimum support itemsets where
	different thresholds are used for finding frequent itemsets for different
	lengths of the itemset.  Finally, we analyze the behavior of our algorithm
	with respect to different parameters and show through experiments that it
	outperforms the competing ones. \\
\end{abstract}


\section{Introduction}
\label{sec:intro}

Mining frequent itemsets has found extensive utilization in various data mining
applications including consumer market-basket analysis~\cite{ragrawal},
inference of patterns from web page access logs~\cite{web}, and iceberg-cube
computation~\cite{iceberg}. Extensive research has, therefore, been conducted in
finding efficient algorithms for frequent itemset mining, especially in
finding association rules~\cite{assocfull}.

However, significantly less attention has been paid to mining of infrequent
itemsets, even though it has got important usage in (i)~mining of negative
association rules from infrequent itemsets~\cite{xindongfull}, (ii)~statistical
disclosure risk assessment where rare patterns in anonymous census data can
lead to statistical disclosure~\cite{minit}, (iii)~fraud detection where rare
patterns in financial or tax data may suggest unusual activity associated with
fraudulent behavior~\cite{minit}, and (iv)~bioinformatics where rare patterns in
microarray data may suggest genetic disorders~\cite{minit}.

The large body of frequent itemset mining algorithms can be broadly classified
into two categories: (i)~candidate generation-and-test paradigm and
(ii)~pattern-growth paradigm. In earlier studies, it has been shown
experimentally that pattern-growth based algorithms are computationally faster
on dense datasets.

Hence, in this paper, we leverage the pattern-growth paradigm
to propose an algorithm IFP\_min for mining minimally infrequent itemsets. For
some datasets, the set of infrequent itemsets can be exponentially large.
Reporting an infrequent itemset which has an infrequent proper subset is
redundant, since the former can be deduced from the latter. Hence, it is
essential to report only the minimally infrequent itemsets.

Haglin et al. proposed an algorithm, MINIT, to mine minimally infrequent
itemsets~\cite{minit}.  It generated all potential candidate minimal infrequent
itemsets using a ranking order of the items based on their supports and then
validated them against the entire database.

Instead, our proposed IFP\_min algorithm proceeds by processing minimally
infrequent itemsets by partitioning the dataset into two parts, one containing a
particular item and the other that does not.

If the support threshold is too high, then less number of frequent itemsets will
be generated resulting in loss of valuable association rules. On the other hand,
when the support threshold is too low, a large number of frequent itemsets and
consequently large number of association rules are generated, thereby making it
difficult for the user to choose the important ones. Part of the problem lies in
the fact that a single threshold is used for generating frequent itemsets
irrespective of the length of the itemset. To alleviate this problem, Multiple
Level Minimum Support (MLMS) model was proposed~\cite{mlms}, where separate
thresholds are assigned to itemsets of different sizes in order to constrain the
number of frequent itemsets mined. This model finds extensive applications in
market basket analysis~\cite{mlms} for optimizing the number of association
rules generated.  We extend our IFP\_min algorithm to the MLMS framework as
well.

In summary, we make the following contributions:
\begin{itemize}
	\item We propose a new algorithm IFP\_min for mining minimally infrequent
		itemsets.  To the best of our knowledge, this is the first such
		algorithm based on pattern-growth paradigm (Section~\ref{sec:algo}).
	\item We introduce the concept of \emph{residual trees} using a variant of
		the FP-tree structure termed as inverse FP-tree (Section~\ref{sec:ifp}).
	\item We propose an optimization on the Apriori algorithm to mine minimally
		infrequent itemsets (Section~\ref{sec:apriori}).
	\item We present a detailed study to quantify the impact of variation in the
		density of datasets on the computation time of Apriori, MINIT and our
		algorithm.
	\item We extend the proposed algorithm to mine frequent itemsets in the MLMS
		framework (Section~\ref{sec:MLMS}).
\end{itemize}

\begin{table}[t]
\centering
{
\begin{tabular}{|c|c|}
\hline
Tid & Transactions\\
\hline
T$_1$ & F, E\\
T$_2$ & A, B, C\\
T$_3$ & A, B\\
T$_4$ & A, D\\
T$_5$ & A, C, D\\
T$_6$ & B, C, D\\
T$_7$ & E, B\\
T$_8$ & E, C\\
T$_9$ & E, D\\
\hline
\end{tabular}
}
\caption{Example database for infrequent itemset mining.}
\label{tab:exmii}
\end{table}

\subsection{Problem Specification}
\label{sec:problem}

Consider $I = \{x_1, x_2, \dots, x_n\}$ to be a set of items. An \emph{itemset}
$X \subseteq I$ is a subset of items. If its length or \emph{cardinality} is
$k$, it is referred to as a \emph{k-itemset}. A \emph{transaction} $T$ is a
tuple $(tid, X)$ where $tid$ is the transaction identifier and $X$ is an
itemset. It is said to \emph{contain} an itemset $Y$ if and only if $Y \subseteq
X$.  A \emph{transaction database} $TD$ is simply a set of transactions.

Each itemset has an associated statistical measure called \emph{support}. For an
itemset $X$, $supp(X, TD) = X.count$ where $X.count$ is the number of
transactions in $TD$ that contains $X$. For a user defined threshold $\sigma$,
an itemset is \emph{frequent} if and only if its support is greater than or
equal to $\sigma$.  It is \emph{infrequent} otherwise.

As mentioned earlier, the number of infrequent itemsets for a particular
database may be quite large. It may be impractical to generate and report all of
them.  A key observation here is the fact that if an itemset is infrequent, so
will be all its supersets. Thus, it makes sense to generate only \emph{minimal}
infrequent itemsets, i.e., those which are infrequent but whose all subsets are
frequent.

\begin{definition}[Minimally Infrequent Itemset]
	An itemset $X$ is said to be \emph{minimally infrequent} for a support
	threshold $\sigma$ if it is infrequent and all its proper subsets are
	frequent, i.e., $supp(X) < \sigma$ and $\forall Y \subset X, supp(Y) \geq
	\sigma$.
\end{definition}

Given a particular support threshold, our goal is to efficiently generate all
the minimally infrequent itemsets (MIIs) using the pattern-growth paradigm. 

Consider an example transaction database shown in Table~\ref{tab:exmii}. If
$\sigma = 2$, $\{B, D\}$ is one of the minimally infrequent itemsets for the
transaction database. All its subsets, i.e., $\{B\}$ and $\{D\}$, are frequent
but it itself is infrequent as its support is $1$.  The whole set of MIIs for
the transaction database is $\{\{E, B\}$, $\{E, C\}$, $\{E, D\}$, $\{B, D\}$,
$\{A, B, C\}$, $\{A, C, D\}$, $\{A, E\}$, $\{F\}\}$. Note that $\{B, F\}$ is not
a MII since one of its subsets $\{F\}$ is infrequent as well.

In the MLMS framework, the problem is to find all frequent (equivalently,
infrequent) itemsets with different support thresholds assigned to itemsets of
different lengths. We define $\sigma_k$ as the minimum support threshold for a
\emph{k-itemset} ($k = 1, 2, \dots, n$) to be frequent.  A \emph{k-itemset} $X$
is frequent if and only if $supp(X,TD) \geq \sigma_k$.

For efficient processing of MLMS itemsets, it is useful to sort the list of
items in each transaction in increasing order of their support counts. This is
called the \emph{i-flist} order.  Also, let $\sigma_{low}$ be lowest minimum
support threshold.

Most applications use the constraint $\sigma_1 \geq \sigma_2 \geq \dots \geq
\sigma_n$.  This is intuitive as the support of larger itemsets can only
decrease or at most remain constant.  In this case, $\sigma_{low} = \sigma_n$.
Our algorithm IFP\_MLMS, however, does not depend on this assumption, and works
for any general $\sigma_1, \dots, \sigma_n$.

Consider an example transaction database shown in Table~\ref{tab:exmlms}. The
frequent itemsets corresponding to the different thresholds are shown in
Table~\ref{tab:freqmlms}.

\begin{table}[t]
\centering
{
\begin{tabular}{|c|c|c|}
\hline
Tid & Transactions & Items in \emph{i-flist} order\\
\hline
T$_1$ & A, C, T, W & A, T, W, C\\
T$_2$ & C, D, W & D, W, C\\
T$_3$ & A, C, T, W & A, T, W, C\\
T$_4$ & A, D, C, W & A, D, W, C\\
T$_5$ & A, T, C, W, D & A, D, T, W,\\
T$_6$ & C, D, T, B & B, D, T, C\\
\hline
\end{tabular}
}
\caption{Example database for MLMS model.}
\label{tab:exmlms}
\end{table}

\begin{table}[t]
\centering
{
\hfill{}
\begin{tabular}{|c|l|}
\hline
$\sigma_k$ & \multicolumn{1}{|c|}{Frequent \emph{k-itemsets}}\\
\hline
$\sigma_1 = 4$ & \{C\}, \{W\}, \{T\}, \{D\}, \{A\} \\
$\sigma_2 = 4$ & \{C, D\}, \{C, W\}, \{C, A\}, \{W, A\}, \{C, T\} \\
\multirow{2}{*}{$\sigma_3 = 3$} & \{C, W, T\}, \{C, W, D\}, \{C, W, A\}, \\
& \quad \{C, T, A\}, \{W, T, A\} \\
$\sigma_4 = 2$ & \{C, W, T, A\}, \{C, D, W, A\} \\
$\sigma_5 = 1$ & \{C, W, T, D, A\} \\
\hline
\end{tabular}
}
\hfill{}
\caption{Frequent \emph{k-itemsets} for database in Table~\ref{tab:exmlms}.}
\label{tab:freqmlms}
\end{table}

\section{Related Work} 
\label{sec:related}

The problem of mining frequent itemsets was first introduced by Agrawal et
al.~\cite{ragrawal}, who proposed the \emph{Apriori}-algorithm. Apriori is a
bottom-up, breadth-first search algorithm that exploits the downward closure
property ``all subsets of frequent itemsets are frequent''.  Only candidate
frequent itemsets whose subsets are all frequent are generated in each database
scan.  Apriori needs $l$ database scans if the size of the largest frequent
itemset is $l$.  In this paper, we propose a variation of the Apriori algorithm
for mining minimally infrequent itemsets (MIIs).

In~\cite{han}, Han et al. introduced a novel algorithm known as the
\emph{FP-growth} method for mining frequent itemsets. The FP-growth method is a
depth-first search algorithm. A data structure called the \emph{FP-tree} is used
for storing the frequency information of itemsets in the original transaction
database in a compressed form.  Only two database scans are needed for the
algorithm and no candidate generation is required. This makes the FP-growth
method much faster than Apriori.  In~\cite{fparray}, Grahne et al. introduced a
novel \emph{FP-array} technique that greatly reduces the need to traverse the
FP-trees. In this paper, we use a variation of the FP-tree for mining the MIIs.

To the best of our knowledge there has been only one other work that discusses
the mining of MIIs. In~\cite{minit}, Haglin et al. proposed the algorithm
\emph{MINIT} which is based upon the \emph{SUDA2} algorithm developed for
finding unique itemsets (itemsets with no unique proper subsets)~\cite{suda1,
suda2}. The authors also showed that the minimal infrequent itemset problem is
NP-complete~\cite{minit}.

In~\cite{mlms}, Dong et al. proposed the MLMS model for constraining the number
of frequent and infrequent itemsets generated. A candidate generation-and-test
based algorithm \emph{Apriori\_MLMS} was proposed in~\cite{mlms}. The downward
closure property is absent in the MLMS model, and thus, the Apriori\_MLMS
algorithm checks the supports of all possible \emph{k-itemsets} occurring at
least once in the transaction database, for finding the frequent itemsets.
Generally, the support thresholds are chosen randomly for different length
itemsets with the constraint $\sigma_i \geq \sigma_j, \forall i < j$.
In~\cite{interestingMLMS}, Dong et al. extended their proposed algorithm
from~\cite{mlms} to include an interestingness parameter while mining frequent
and infrequent itemsets.

\section{Need for Residual Trees}
\label{sec:why}

By definition, an itemset is a minimally infrequent itemset (MII) if and only if
it is infrequent and all its subsets are frequent. Thus, a trivial algorithm to
mine all MIIs would be to compute all the subsets for every infrequent itemset
and check if they are frequent. This involves finding all the frequent and
infrequent itemsets in the database and proceeding with checking the subsets of
infrequent itemsets for occurrence in the large set of frequent itemsets. This
is a simple but computationally expensive algorithm.

The use of \emph{residual trees} reduces the computation time.  A residual tree
for a particular item is a tree representation of the residual database
corresponding to the item, i.e., the entire database with the item removed.  We
show later that a MII found in the residual tree is a MII of the entire
transaction database.

The projected database, on the other hand, corresponds to the set of
transactions that contains a particular item.  A potential minimal infrequent
itemset mined from the projected tree must not have any infrequent subset.  The
itemset itself is a subset since it is actually the union with the item of the
projected tree that is under consideration.  As we show later, the support of
only this itemset needs to be computed from the corresponding residual tree.

In this paper, our proposed algorithm \emph{IFP\_min} uses a structure similar
to the FP-tree~\cite{han} called the \emph{IFP-tree}.  This is due to the fact
that the IFP-tree provides a more visually simplified version of the residual
and projected trees that leads to enhanced understanding of the algorithm.  A
similar algorithm FP\_min can be designed that uses the FP-tree.  The time
complexity remains the same. In the next section, we describe in detail the
IFP-tree and the corresponding structures, projected tree and residual tree.

\section{Inverse FP-tree (IFP-tree)}
\label{sec:ifp}

The \emph{inverse FP-tree (IFP-tree)} is a variant of the FP-tree~\cite{han}.
It is a compressed representation of the whole transaction database.  Every path
from the root to a node represents a transaction. The root has an empty label.
Except the root node, each node of the tree contains four fields: (i)~item id,
(ii)~count, (iii)~list of child links, and (iv)~a node link. The item id field
contains the identifier of the item. The count field at each node stores the
support count of the path from the root to that node. The list of child links
point to the children of the node. The node link field points to another node
with the same item id that is present in some other branch of the tree. 

An \emph{item header table} is used to facilitate the traversal of the tree. The
header table consists of items and a link field associated with each item that
points to the first occurrence of the item in the IFP-tree. All link entries of
the header table are initially set to \emph{null}. Whenever an item is added
into the tree, the corresponding link entry of the header table is updated.
Items in each transaction are sorted according to their order in \emph{i-flist}.

For inserting a transaction, a path that shares the same prefix is searched. If
there exists such a path, then the count of the common prefix is incremented by
one in the tree and the remaining items of the transaction (which do not share
the path) are attached from the last node with their count value set to $1$. If
items of a transaction do not share any path in the tree, then they are attached
from the root. The IFP-tree for the sample transaction database in
Table~\ref{tab:exmii} (considering only the transactions T$_2$ to T$_9$) is shown in
Figure~\ref{fig:full_ifp}.

\begin{figure}[t]
\centering
\includegraphics[width=0.70\columnwidth]{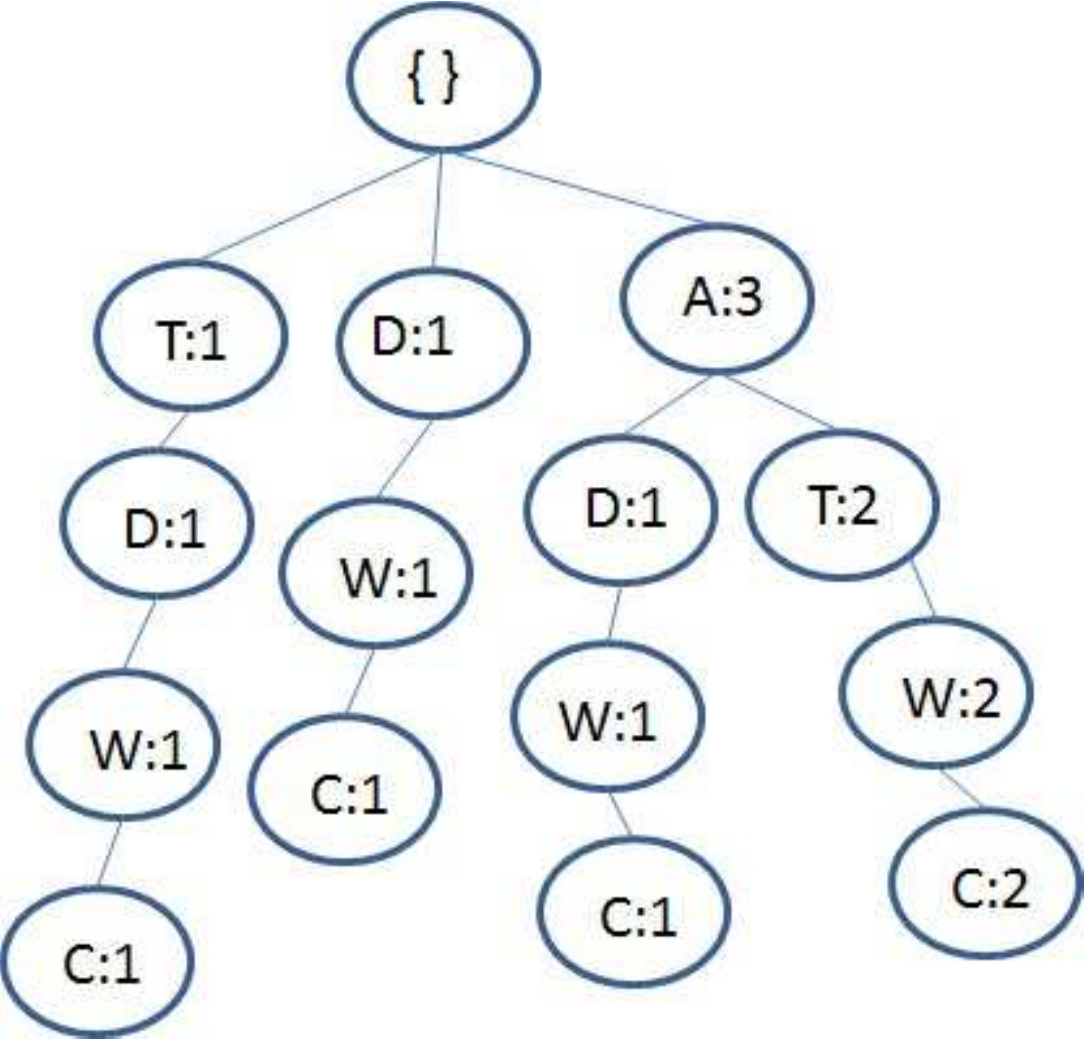}
\caption{IFP-tree corresponding to the transaction database in
Table~\ref{tab:exmii} (T$_2$-T$_9$).}
\label{fig:full_ifp}
\end{figure}

The \emph{IFP\_min} algorithm recursively mines the minimally infrequent
itemsets (MIIs) by dividing the IFP-tree into two sub-trees: projected tree and
residual tree.  The next two sections describe them.

\subsection{Projected Tree}
\label{sec:proj}

Suppose $TD$ be the transaction database represented by the IFP-tree $T$ and
$TD_x$ denotes the database of transactions that contain the item $x$. The
\emph{projected database} corresponding to the item $x$ is the database of these
transactions $TD_x$, but after removing the item $x$ from the transactions. The
IFP-tree corresponding to this database is called the \emph{projected tree}
$T_{P_x}$ of item $x$ in $T$.  Figure~\ref{fig:proj1} shows the projected tree
of item $A$, which is the least frequent item in Figure~\ref{fig:full_ifp}.

The IFP\_min algorithm considers the projected tree of only the \emph{least
frequent item (lf-item)}.  Henceforth, for simplicity, we associate every
IFP-tree with only a single projected tree which is that of the \emph{lf-item}.
Moreover, since the items are sorted in the \emph{i-flist} order, there would be
a single node of the \emph{lf-item} $x$ in the IFP-tree. Thus, the projected
tree of $x$ can be obtained directly from the IFP-tree by considering the
subtree rooted at node $x$.

\subsection{Residual Tree}
\label{sec:resd}

The \emph{residual database} corresponding to an item $x$ is the database of
transactions obtained by removing item $x$ from $TD$. The IFP-tree corresponding
to this database is called \emph{residual tree} $T_{R_x}$ of item $x$ in $T$.
Figure~\ref{fig:resd1} shows the residual tree of the least frequent item $A$.
It is obtained by deleting the node corresponding to item $A$ and then merging
the subtree below that node into the main tree at appropriate positions.

Similar to the projected tree, the IFP\_min algorithm considers the residual
tree of only the lf-item.  Since there is only a single node of the lf-item $x$
in the IFP-tree, the residual tree of $x$ can be obtained directly from the
IFP-tree by deleting the node $x$ from the tree and then merging the subtree
below the node $x$ with the rest of the tree. 

Furthermore, the projected and residual tree of the next item (i.e., $E$) in the
i-flist is associated with the residual tree of the current item (i.e., $A$).
Figure~\ref{fig:projresd2} shows the projected and residual trees of the item
$E$ for the tree $T_{R_A}$.

\begin{figure}[t]
\centerline{
\subfloat[]{\includegraphics[width=0.3\columnwidth]{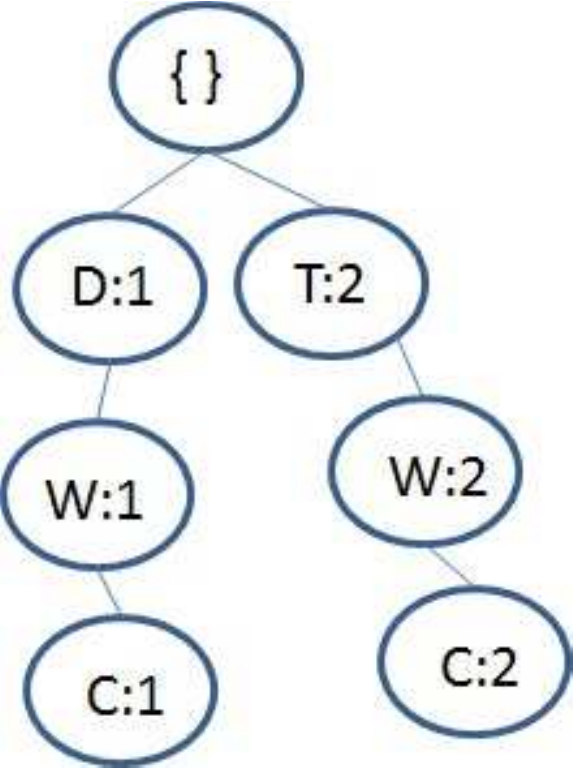}
\label{fig:proj1}}
\qquad
\subfloat[]{\includegraphics[width=0.6\columnwidth]{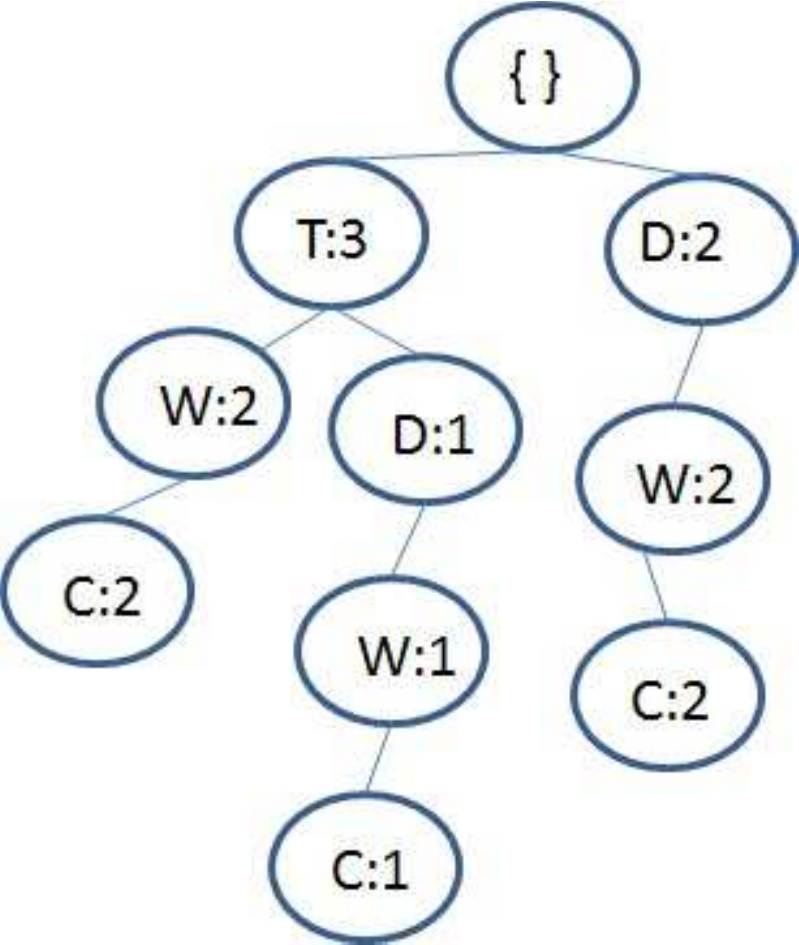}
\label{fig:resd1}}
}
\caption{(a) Projected tree $T_{P_A}$ and (b) Residual tree $T_{R_A}$ of item
$A$ for the IFP-tree shown in Figure~\ref{fig:full_ifp}.}
\label{fig:projresd1}
\end{figure}

\begin{figure}[t]
\centerline{
\subfloat[]{\includegraphics[width=0.3\columnwidth]{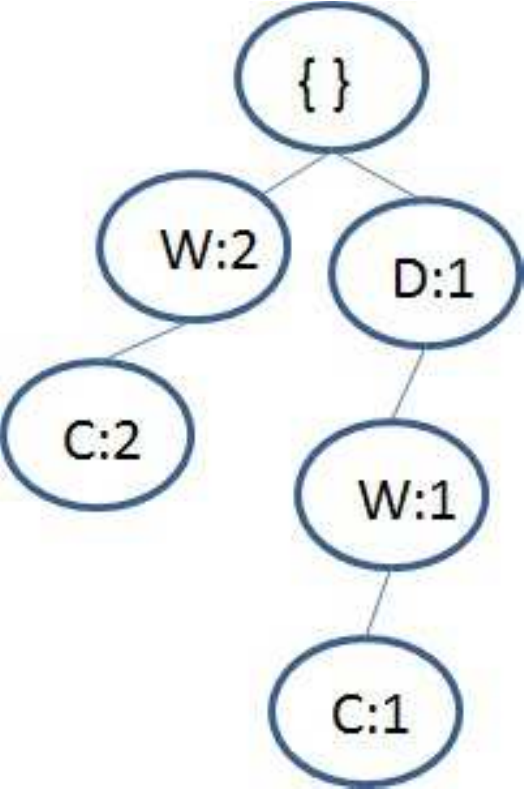}
\label{fig:proj2}}
\qquad
\subfloat[]{\includegraphics[width=0.3\columnwidth]{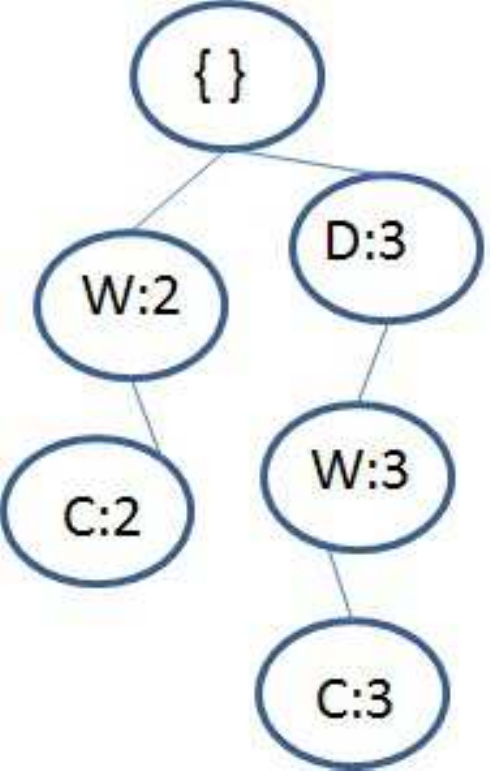}
\label{fig:resd2}}
}
\caption{(a) Projected tree and (b) Residual tree of
item $E$ in the residual tree $T_{R_A}$ shown in Figure~\ref{fig:resd1}.}
\label{fig:projresd2}
\end{figure}

\begin{algorithm}[t]
\caption{IFP\_min}
\label{algo:ifp}
\begin{algorithmic}[1]
\REQUIRE $T$ is an IFP-tree\\
\ENSURE Minimally infrequent itemsets (MIIs) in $db(T)$
\STATE $infreq(T) \leftarrow$ infrequent 1-itemsets
\STATE $T \leftarrow$ IFP-tree($db(T) - infreq(T)$)
\STATE $x \leftarrow$ least frequent item in $T$ (the \emph{lf-item})
\IF {$T$ has a single node} 
	\IF {$\{x\}$ is infrequent in $T$}
		\RETURN {$\{x\}$}
	\ELSE
  		\RETURN {$\varnothing$}
	\ENDIF
\ENDIF				 		
\STATE $T_{P_x} \leftarrow$ projected tree of $x$
\STATE $T_{R_x} \leftarrow$ residual tree of $x$
\STATE $S_{R} \leftarrow$ IFP\_min (residual tree $T_{R_x})$
\STATE $S_{R_{infreq}} \leftarrow$ $S_{R}$
\STATE $S_{P} \leftarrow$ IFP\_min (projected tree $T_{P_x})$
\STATE $S_{P_{infreq}} \leftarrow$ $\{x\} \bullet (S_{P} - S_{R})$
\STATE $S_{2}(x) \leftarrow$ $\{x\} \bullet \left(items(T_{R_x}) - items(T_{P_x})\right)$
\RETURN $\{S_{R_{infreq}} \cup S_{P_{infreq}} \cup S_{2}(x) \cup infreq(T)\}$
\end{algorithmic}
\end{algorithm}

\section{Mining Minimally Infrequent Itemsets}
\label{sec:algo}

In this section, we describe the IFP\_min algorithm that uses a recursive
approach to mine minimally infrequent itemsets (MIIs).  The steps of the
algorithm are shown in Algorithm~\ref{algo:ifp}.

The algorithm uses a dot-operation ($\bullet$) that is used to unify sets.  The
unification of the first set with the second produces another set whose
$i^\text{th}$ element is the union of the entire first set with the
$i^\text{th}$ element of the second set.  Mathematically, $\{x\} \bullet \{S_1,
\dots, S_n\}$ $= \{\{x\} \cup S_1, \dots, \{x\} \cup S_n\}$ and $\{x\} \bullet
\varnothing = \varnothing$.

The infrequent 1-itemsets are trivial MIIs and so they are reported and pruned
from the database in Step 1.  After this step, all the items present in the
modified database are individually frequent.  IFP\_min then selects the least
frequent item (lf-item, Step 3) and divides the database into two non-disjoint
sets: projected database and residual database of the lf-item (Step 11 and Step
12 respectively).

The IFP\_min algorithm is then applied to the residual database in Step 13 and
the corresponding MIIs are reported in Step 14. In the base case (Step 6 to Step
12) when the residual database consists of a single item, the MII reported is
either the item itself or the empty set accordingly as the item is infrequent or
frequent respectively.

After processing the residual database, the algorithm mines the projected
database in Step 15. The itemsets in the projected database share the lf-item as
a prefix. The MIIs obtained from the projected database by recursively applying
the algorithm are compared with those obtained from residual database. If an
itemset is found to occur in the second set, it is not reported; otherwise, the
lf-item is included in the itemset and is reported as an MII of the original
database (Step 16).

IFP\_min also reports the 2-itemsets consisting of the lf-item and frequent
items not present in the projected database of the lf-item (Step 17). These
2-itemsets have support zero in the actual database and hence also qualify as
MIIs.

\subsection{Example}
\label{sec:example}

Consider the transaction database $TD$ shown in Table~\ref{tab:exmii}.
Figure~\ref{fig:example2} shows the recursive partitioning of the tree $T$
corresponding to $TD$. The box associated with each tree represents the MIIs in
that tree for $\sigma = 2$.

The algorithm starts by separating the infrequent items from the database. This
results in removal of the item $F$ (which is inherently a MII). The lf-item in
the modified tree $T$ is $A$. The algorithm then constructs the projected tree
$T_{P_A}$ and the residual tree $T_{R_A}$ corresponding to item $A$ and
recursively processes them to yield MIIs containing $A$ and MIIs not containing
$A$ respectively.

\begin{itemize}

	\item MIIs not containing $A$ (itemsets obtained from $T_{R_A}$)
	
		The lf-item in $T_{R_A}$ is $E$. Therefore, similar to the first step,
		$T_{R_A}$ is again divided into projected tree $T_{P_E}$ and residual
		tree $T_{R_E}$, which are then recursively mined.

	\begin{itemize}
		
		\item MIIs not containing $E$ (itemsets obtained from $T_{R_E}$)

			Every 1-itemset is frequent.  By recursively processing the residual
			and projected trees of $B$, $\{B, D\}$ is obtained as a MII.
			Itemset $\{C, D\}$ is also obtained as a potential MII.  However,
			since it is frequent, it is not returned.  Since $\{B, C\}$ is also
			frequent, only $\{B, D\}$ is returned from this tree.
	
		\item MIIs containing $E$ (itemsets obtained from $T_{P_E}$)

			All the 1-itemsets $\{B\}, \{C\}, \{D\}$ are infrequent (Step 6 of
			the algorithm).  $E$ is included with these itemsets to return $\{E,
			B\}, \{E, C\}, \{E, D\}$.

	\end{itemize}

		$T_{P_E}$ and $T_{R_E}$ are mutually exclusive. Hence, the combined set
		$\{\{B, D\}, \{E, B\}, \{E, C\}, \{E, D\}\}$ forms the MIIs not
		containing $A$.

	\item MIIs containing $A$ (itemsets obtained from $T_{P_A}$)

		Item $\{E\}$ is infrequent ($support = 0$).  Hence, $\{A, E\}\}$ forms a
		MII.  Similarly, itemset $\{B, D\}$ with $support = 0$ is also obtained
		as a potential MII.  The other MIIs obtained from recursive processing
		are $\{B, C\}, \{C, D\}$.  The itemset $\{B, D\}$, however, appears as a
		MII in both $T_{P_A}$ and $T_{R_A}$.  Hence, it is removed (Step 16 of
		the algorithm). This avoids the inclusion of the itemset $\{A, B, D\}$
		which is not a MII since $\{B, D\}$ is infrequent (as shown in
		$T_{R_A}$).  $A$ is included with the remaining set of MIIs from
		$T_{P_A}$ to form the itemsets $\{A, B, C\}, \{A, C, D\}$.
		
		The combined set $\{\{A, B, C\}, \{A, C, D\}, \{A, E\}\}$ thus forms the
		MIIs containing $A$.

\end{itemize}

As mentioned in Step 3 of the algorithm, the infrequent 1-itemset $\{F\}$ is
also included.  Hence, in all, the algorithm returns the set $\{\{B, D\},$ $\{E,
B\},$ $\{E, C\},$ $\{E, D\},$ $\{A, B, C\},$ $\{A, C, D\},$ $\{A, E\},$
$\{F\}\}$ as the MIIs for the database $TD$.

\begin{figure*}
\centering
\includegraphics[width=\textwidth]{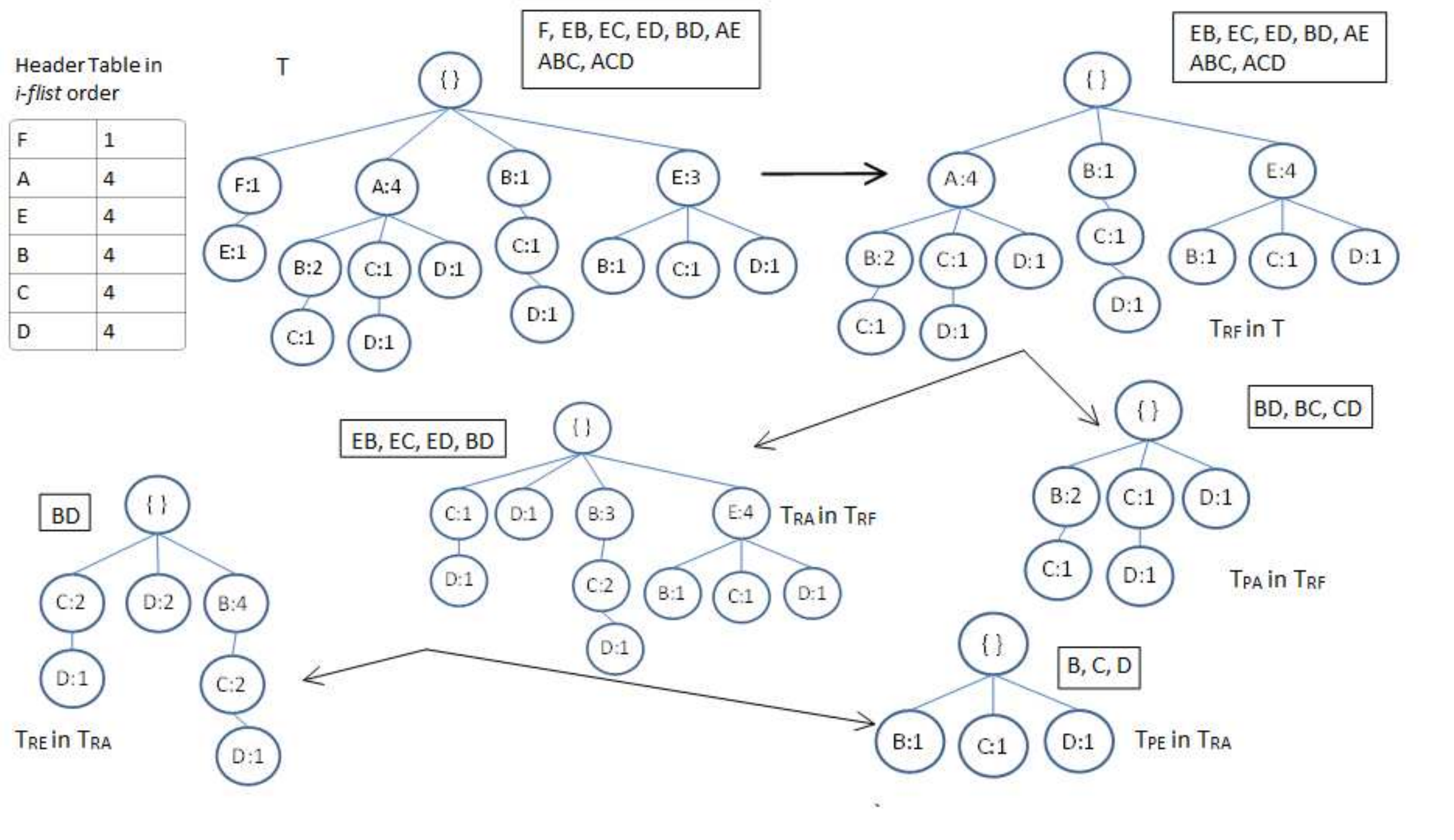}
\caption{Running example for the transaction database in Table~\ref{tab:exmii}.}
\label{fig:example2}
\end{figure*}

\subsection{Completeness and Soundness}
\label{sec:explain}

In this section, we prove formally the our algorithm is \emph{complete} and
\emph{sound}, i.e., it returns all minimally infrequent itemsets and all
itemsets returned by it are minimally infrequent.

Consider the lf-item $x$ in the transaction database. MIIs in the database can
be exhaustively divided into the following sets:

\begin{itemize}

	\item Group 1: MIIs not containing $x$

		This can be again of two kinds:

		\begin{enumerate}[(a)]

			\item Itemsets of length 1:
				
				These constitute the infrequent items in the database obtained
				from the initial pruning.

			\item Itemsets of length greater than 1:
				
				These constitute the minimally infrequent itemsets obtained from
				the residual tree of $x$.

		\end{enumerate}

	\item Group 2: MIIs containing $x$

		This consists of itemsets of the form $\{x\} \cup S$ where $S$ can be of
		following two kinds:

		\begin{enumerate}[(a)]

			\item $S$ occurs with $x$ in at least one transaction:
			
				All items in $S$ occur in the projected tree of $x$.

			\item $S$ does not occur with $x$ in any transaction:

				Note that the path corresponding to $S$ does not exist in the
				projected tree of $x$.  Also, $S$ is a 1-itemset.  Assume the
				contrary, i.e., $S$ contains two or more items.  A subset of $S$
				would then exist, which would not occur with $x$ in any
				transaction.  As a result, the subset would also be infrequent
				and $\{x\} \cup S$ would not be qualified as a MII.  Thus, $S$
				is a node that is absent in the projected tree of $x$.

		\end{enumerate}

\end{itemize}

The following set of observations and theorems prove the correctness of the
IFP\_min algorithm. They exhaustively define and verify the set of itemsets
returned by the algorithm.

The first observation relates the (in)frequent itemsets of the database to those
present in the residual database.

\begin{observation}
	\label{obs:rfreq}
	An itemset $S$ (not containing the item $x$) is frequent (infrequent) in the
	residual tree $T_{R_x}$ if and only if the itemset $S$ is frequent
	(infrequent) in $T$, i.e.,
	\begin{center}
		$S$ is frequent in $T$ $\Leftrightarrow$ $S$ is frequent
		in $T_{R_x}$ ($x \notin S$) \\
		$S$ is infrequent in $T$ $\Leftrightarrow$ $S$ is 
		infrequent in $T_{R_x}$ ($x \notin S$)
	\end{center}
\end{observation}

\begin{proof}
	Since $S$ does not contain $x$, it does not occur in the projected tree of
	$x$. All occurrences of $S$ must, therefore, be only in the residual tree
	$T_{R_x}$, i.e.,
	\begin{align}
		\label{eq:1}
		supp(S, T) = supp(S, T_{R_x}) \quad [x \notin S]
	\end{align}
	Hence, $S$ is frequent (infrequent) in $T$ if and only if $S$ is frequent
	(infrequent) in $T_{R_x}$.
\end{proof}

The following theorem shows that the MIIs that do not contain the item $x$ can
be obtained directly as MIIs from the residual tree of $x$.

\begin{theorem}
	\label{thm:resdMII}
	An itemset $S$ (not containing the item $x$) is minimally infrequent in $T$
	if and only if the itemset $S$ is minimally infrequent in the residual tree
	$T_{R_x}$ of $x$, i.e.,
	\begin{center}
		$S$ is minimally infrequent in $T$ \\
		$\Leftrightarrow$ $S$ is minimally infrequent in $T_{R_x}$ ($x \notin
		S$)
	\end{center}			
\end{theorem}

\begin{proof}
	Suppose $S$ is minimally infrequent in $T$.  Therefore, it is itself
	infrequent, but all its subsets $S' \subset S$ are frequent.

	As $S$ does not contain $x$, all occurrences of $S$ or any of its subsets
	$S' \subset S$ must occur in the residual tree $T_{R_x}$ only.  Hence, using
	Observation~\ref{obs:rfreq}, in $T_{R_x}$ also, $S$ is infrequent but all
	$S' \subset S$ are frequent.  Therefore, $S$ is minimally infrequent in
	$T_{R_x}$ as well.

	The converse is also true, i.e., if $S$ is minimally infrequent in
	$T_{R_x}$, since all its occurrences are in $T_{R_x}$ only, it is globally
	minimally infrequent (i.e., in $T$) as well.
\end{proof}

Theorem~\ref{thm:resdMII} guides the algorithm in mining MIIs not containing the
least frequent item (lf-item) $x$.  The algorithm makes use of this theorem
recursively, by reporting MIIs of residual trees as MIIs of the original tree.
For mining of MIIs that contain $x$, the following observation and theorem are
presented.

The second observation relates the (in)frequent itemsets of the database to
those present in the projected database.

\begin{observation}
	\label{obs:pfreq}
	An itemset $S$ is frequent (infrequent) in the projected tree $T_{P_x}$ if
	and only if the itemset obtained by including $x$ in $S$ (i.e., $x \cup S$),
	is frequent (infrequent) in $T$, i.e.,
	\begin{center}
		$x \cup S$ is frequent in $T$ $\Leftrightarrow$ $S$ is frequent in
		$T_{P_x}$ \\
		$x \cup S$ is infrequent in $T$ $\Leftrightarrow$ $S$ is infrequent in
		$T_{P_x}$
	\end{center}
\end{observation}

\begin{proof}
	Consider the itemset $x \cup S$.  All occurrences of it are only in the
	projected tree $T_{P_x}$.  The projected tree, however, does not list $x$,
	and therefore, we have,
	\begin{align}
		\label{eq:2}
		supp(x \cup S, T) = supp(x \cup S, T_{P_x}) = supp(S, T_{P_x})
	\end{align}
	Hence, $x \cup S$ is frequent (infrequent) in $T$ if and only if $S$ is
	frequent (infrequent) in $T_{P_x}$.
\end{proof}

The next theorem shows that the potential MIIs obtained from the projected tree
of an item $x$ (by appending $x$ to it) is a MII provided it is not a MII in the
corresponding residual tree of $x$.

\begin{theorem}
	\label{thm:projMII}
	An itemset $\{x\} \cup S$ is minimally infrequent in $T$ if and only if the
	itemset $S$ is minimally infrequent in the projected tree $T_{P_x}$ but not
	minimally infrequent in the residual tree $T_{R_x}$, i.e.,
	\begin{center}
		$\{x\} \cup S$ is minimally infrequent in $T$ \\
		$\Leftrightarrow$ $S$ is minimally infrequent in $T_{P_x}$ and \\
		$S$ is not minimally infrequent in $T_{R_x}$
	\end{center}			
\end{theorem}

\begin{proof}

	\emph{LHS to RHS:}

	Suppose $\{x\} \cup S$ is minimally infrequent in $T$.  Therefore, it is
	itself infrequent, but all its subsets $S' \subset \{x\} \cup S$, including
	$S$, are frequent.

	Since $S$ is frequent and it does not contain $x$, using
	Observation~\ref{obs:rfreq}, $S$ is frequent in $T_{R_x}$ and is, therefore,
	not minimally infrequent in $T_{R_x}$.

	From Observation~\ref{obs:pfreq}, $S$ is infrequent in $T_{P_x}$.  Assume
	that $S$ is not minimally infrequent in $T_{P_x}$.  Since it is infrequent
	itself, there must exist a subset $S' \subset S$ which is infrequent in
	$T_{P_x}$ as well.  Now, consider the itemset $\{x\} \cup S'$.  From
	Observation~\ref{obs:pfreq}, it must be infrequent in $T$.  However, since
	this is a subset of $\{x\} \cup S$, this contradicts the fact that $\{x\}
	\cup S$ is minimally infrequent.  Therefore, the assumption that $S$ is not
	minimally infrequent in $T_{P_x}$ is false.

	Together, it shows that if $\{x\} \cup S$ is minimally infrequent in $T$,
	then $S$ is minimally infrequent in $T_{P_x}$ but not in $T_{R_x}$.

	\emph{RHS to LHS:}

	Given that $S$ is minimally infrequent in $T_{P_x}$ but not in $T_{R_x}$,
	assume that $\{x\} \cup S$ is not minimally infrequent in $T$.  Since $S$ is
	infrequent in $T_{P_x}$, using Observation~\ref{obs:pfreq}, $\{x\} \cup S$
	is also infrequent in $T$.
	
	Now, since we have assumed that $\{x\} \cup S$ is not minimally infrequent
	in $T$, it must contain a subset $A \subset \{x\} \cup S$ which is
	infrequent in $T$ as well.

	Suppose $A$ contains $x$, i.e., $x \in A$.  Consider the itemset $B$ such
	that $A = \{x\} \cup B$.  Note that since $A \subset \{x\} \cup S$, $B
	\subset S$.  Since $A = \{x\} \cup B$ is infrequent in $T$, from
	Observation~\ref{obs:pfreq}, $B$ is infrequent in $T_{P_x}$.  However, since
	$B \subset S$, this contradicts the fact that $S$ is minimally infrequent in
	$T_{P_x}$.  Hence, $A$ cannot contain $x$.

	Therefore, it must be the case that $x \notin A$.  To show that this leads
	to a contradiction as well, we first show that
	
	Now, if $S$ is minimally infrequent in $T_{P_x}$ but not in $T_{R_x}$ and
	$\{x\} \cup S$ is not minimally infrequent in $T$, then every subset $S'
	\subset S$ is frequent in $T_{P_x}$.  Thus, $\forall S' \subset S, S'$ is
	frequent in $T$ as well (since $T_{P_x}$ is only a part of $T$).  Then, if
	$S$ is infrequent, it must be a MII.  However, from
	Theorem~\ref{thm:resdMII}, it becomes a MII in $T_{R_x}$ which is a
	contradiction.  Therefore, $S$ cannot be infrequent and is, therefore,
	frequent in $T$.

	Since, $A \subset \{x\} \cup S$ but $x \notin A$, therefore, $A \subset S$.
	We have already shown that $A$ is infrequent in $T$.  Using the Apriori
	property, $S$ cannot then be frequent.  Hence, it contradicts the original
	assumption that $\{x\} \cup S$ is not minimally infrequent in $T$.

	Together, we get that if $S$ is minimally infrequent in $T_{P_x}$ but not in
	$T_{R_x}$, then $\{x\} \cup S$ is minimally infrequent in $T$.
\end{proof}

Theorem~\ref{thm:projMII} guides the algorithm in mining MIIs containing the
least frequent item (lf-item) $x$.  The algorithm first obtains the MIIs from
its projected tree and then removes those that are also found in the residual
tree.  It thus shows the connection between the two parts of the database,
projected and residual.

\subsection{Correctness}

We now formally establish the correctness of the algorithm IFP\_min by showing
that the MIIs as enumerated in Group 1 and Group 2 in Section~\ref{sec:explain}
are generated by it.

In Step 1, the algorithm first finds the infrequent items present in the tree.
These 1-itemsets cover the Group 1(a).

Consider the least frequent item $x$.  In Step 11 and Step 12, the tree $T$ is
divided into smaller trees, the residual tree $T_{R_x}$ and the projected tree
$T_{P_x}$.

In Step 13, Group 1(b) MIIs are obtained by the recursive application of
IFP\_min on the residual tree $T_{R_x}$.  Theorem~\ref{thm:resdMII} proves that
these are MIIs in the original dataset.

In Step 14, potential MIIs are obtained by the recursive application of IFP\_min
on the projected tree $T_{P_x}$.  MIIs obtained from $T_{R_x}$ are removed from
this set.  Combined with the item $x$, these form the MIIs enumerated as Group
2(a).  Theorem~\ref{thm:projMII} proves that these are indeed MIIs in the
original dataset.

The projected database consist of all those transactions in which $x$ is
present.  The Group 2(b) MIIs are of length 2 (as shown earlier).  Thus, single
items that are frequent but do not appear in the projected tree of $x$, when
combined with $x$, constitute MIIs with support count of zero.  These items
appear in $T_{R_x}$ though as they are frequent.  Hence, they are obtained as
single items that appear in $T_{R_x}$ but not in $T_{P_x}$ as shown in Step 17
of the algorithm.

The algorithm is, hence, complete and it exhaustively generates all minimally
infrequent itemsets.

\subsection{MIIs using Apriori}
\label{sec:apriori}

In this section, we show how the Apriori algorithm~\cite{ragrawal} can be
improved to mine the MIIs.  Consider the iteration where candidate itemsets of
length $l+1$ are generated from frequent itemsets of length $l$.  From the
generated candidate set, itemsets whose support satisfies the minimum support
threshold are reported as frequent and the rest are rejected.  This rejected set
of itemsets constitute the MIIs of length $l+1$.  This is attributed to the fact
that for such an itemset, all the subsets are frequent (due to the candidate
generation procedure) while the itemset itself is infrequent.  For the
experimentation purposes, we label this algorithm as the \emph{Apriori\_min}
algorithm.

\section{Frequent Itemsets in MLMS Model}
\label{sec:MLMS}

\begin{algorithm}[t]
\caption{IFP\_MLMS}
\label{algo:mlms}
\begin{algorithmic}[1]
\REQUIRE IFP-tree $T$ with $\rho_T = p$\\
\ENSURE Frequent* itemsets of $T$
\IF{$T = \varnothing$} 
	\RETURN $\varnothing$
\ENDIF				 		
\STATE $x \leftarrow$ ls-item in $T$
\IF{$supp(\{x\}, T) < \sigma_{low}$} 
	\STATE $S_P \leftarrow \varnothing$
\ELSE
	\STATE $T_{P_x} \leftarrow$ projected tree of $x$
	\STATE $\rho_{T_{P_x}} \leftarrow p+1$
	\STATE $S_P \leftarrow$ IFP\_MLMS ($T_{P_x})$
\ENDIF				 		
\STATE $T_{R_x} \leftarrow$ residual tree of $x$
\STATE $\rho_{T_{R_x}} \leftarrow p$
\STATE $S_R \leftarrow$ IFP\_MLMS ($T_{R_x})$
\IF{$x$ is \emph{frequent}* in $T$}
	\RETURN $\left(\{x\} \bullet S_P\right) \cup S_R \cup \{x\}$
\ELSE
	\RETURN $\left(\{x\} \bullet S_P\right) \cup S_R$
\ENDIF
\end{algorithmic}
\end{algorithm}

In this section, we present our proposed algorithm \emph{IFP\_MLMS} to mine
frequent itemsets in the MLMS framework.  Though most applications use the
constraint $\sigma_1 \geq \sigma_2 \geq \dots \geq \sigma_n$ for the different
support thresholds at different lengths, our algorithm (shown in
Algorithm~\ref{algo:mlms}) does not depend on it and works for any general
$\sigma_1, \dots, \sigma_n$.

We use the lowest support $\sigma_{low} = \min_{\forall i} \sigma_i$ in the
algorithm.  The algorithm is based on the item with \emph{least support} which
we term as the \emph{ls-item}.

IFP\_MLMS is again based on the concepts of residual and projected trees.  It
mines the frequent itemsets by first dividing the database into projected and
residual trees for the ls-item $x$, and then mining them recursively.  We show
that the frequent itemsets obtained from the residual tree are frequent itemsets
in the original database as well.

The itemsets obtained from the projected tree share the ls-item as a prefix.
Hence, the thresholds cannot be applied directly as the length of the itemset
changes.  The prefix accumulates as the algorithm goes deeper into recursion,
and hence, a track of the prefix is maintained at each recursion level.  At any
stage of the recursive processing, if $supp(ls$-$item) < \sigma_{low}$, then
this item cannot occur in a frequent itemset of any length (as any superset of
it will not pass the support threshold).  Thus, its sub-tree is pruned, thereby
reducing the search space considerably.

\begin{figure*}[t]
\centering
\includegraphics[width=1.0\textwidth]{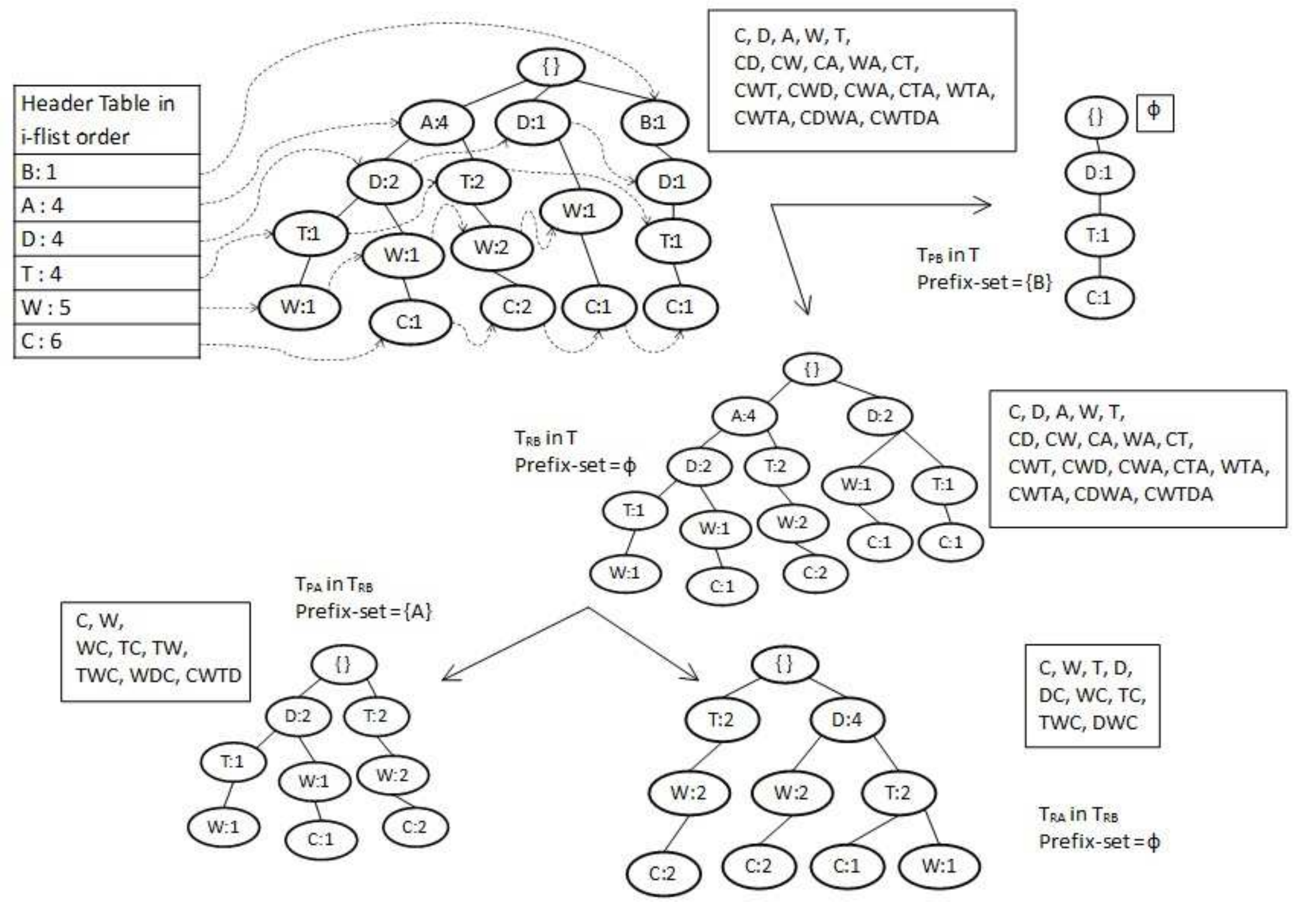}
\caption{Frequent itemsets generated from the database represented in
Table~\ref{tab:exmlms} using the MLMS framework. The box associated with each
tree represents the frequent* itemsets in that tree.}
\label{fig:all_freq}
\end{figure*}

To analyze the prefix items corresponding to a tree, the following two
definitions are required:

\begin{definition}[Prefix-set of a tree]
	The \emph{prefix-set} of a tree is the set of items that need to be included
	with the itemsets in the tree, i.e., all the items on which the projections
	have been done.  For a tree $T$, it is denoted by $\Delta_T$.
\end{definition}

\begin{definition}[Prefix-length of a tree]
	The \emph{prefix-length} of a tree is the length of its prefix-set.  For a
	tree $T$, it is denoted by $\rho_T$.
\end{definition}

For a tree $T$ having $\Delta_T = S$ and $\rho_T = p$, the corresponding values
for the residual tree are $\Delta_{T_{R_x}} = S$ and $\rho_{T_{R_x}} = p$ while
those for the projected tree are $\Delta_{T_{P_x}} = \{x\} \cup S$ and
$\rho_{T_{P_x}} = p+1$.  For the original transaction database $TD$ and its tree
$T$, $\Delta_T = \varnothing$ and $\rho_T = 0$.

For any $k$-itemset $S$ in a tree $T$, the original itemset must include all the
items in the prefix-set.  Therefore, if $\rho_T = p$, for $S$ to be frequent, it
must satisfy the support threshold for $k+p$-length itemsets, i.e., $supp(X)
\geq \sigma_{k+p}$ (henceforth, $\sigma_{k,p}$ is also used to denote
$\sigma_{k+p}$).  The definitions of frequent and infrequent itemsets in a tree
$T$ with a prefix-length $\rho_T = p$ are, thus, modified as follows.

\begin{definition}[Frequent* itemset]
	A k-itemset $S$ is \emph{frequent}* in $T$ having $\rho_T = p$ if $supp(S,
	T) \geq \sigma_{k,p}$.
\end{definition}

\begin{definition}[Infrequent* itemset]
	A k-itemset $S$ is \emph{infrequent}* in $T$ having $\rho_T = p$ if $supp(S,
	T) < \sigma_{k,p}$.
\end{definition}

Using these definitions, we explain Algorithm~\ref{algo:mlms} along with an
example shown in Figure~\ref{fig:all_freq} for the database in
Table~\ref{tab:exmlms}.
The support thresholds are
$\sigma_1 = 4$, $\sigma_2 = 4$, $\sigma_3 = 3$, $\sigma_4 = 2$, and $\sigma_5 = 1$.

The item with least support for the initial tree $T$ is $B$.  Since its support
is above $\sigma_{low}$ (Step 5), the algorithm extracts the projected tree
$T_{P_B}$ and then recursively processes it (Step 11 to Step 13).

The algorithm then processes the residual tree $T_{R_B}$ recursively (Step 12 to
Step 14).  For that, it breaks it into the projected and residual trees of the
ls-item there, which is $A$.  Figure~\ref{fig:all_freq} shows all the frequent
itemsets mined.

Since $B$ itself is not frequent in $T$ (Step 15), it is not returned.  The
other itemsets are returned (Step 18).

\subsection{Correctness of IFP\_MLMS}
\label{mlmsComplete}

Let $x$ be the ls-item in tree $T$ (having $\rho_T = p$) and $S$ be a
\emph{k-itemset} not containing $x$.

For computing the \emph{frequent}* itemsets for a tree $T$, the IFP\_MLMS
algorithm merges the \emph{frequent}* itemsets, obtained by the processing of
the projected tree and residual tree of ls-item, using the following theorem 

\begin{theorem}
	\label{thm:pmlms}
	An itemset $\{x\} \cup S$ is frequent* in $T$ if and only if $S$ is
	frequent* in the projected tree $T_{P_x}$ of $x$, i.e.,
	\begin{center}
		$\{x\}\cup S$ is frequent* in $T \Leftrightarrow S$ is frequent* in
		$T_{P_x}$
	\end{center}
\end{theorem}

\begin{proof}
	Suppose $S$ is a $k$-itemset. \\
	$S$ is \emph{frequent}* in $T_{P_x}$ \\
	$\Leftrightarrow$ $supp(S, T_{P_x}) \geq \sigma_{k,p+1}$ (since
	$\rho_{T_{P_x}} = p+1$) \\
	$\Leftrightarrow$ $supp(\{x\} \cup S, T) \geq \sigma_{k,p+1}$ (using
	Observation~\ref{obs:pfreq}) \\
	$\Leftrightarrow$ $supp(\{x\} \cup S, T) \geq \sigma_{k+1,p}$ \\
	$\Leftrightarrow$ $\{x\} \cup S$ is \emph{frequent}* in $T$ (since $\rho_{T}
	= p$).
\end{proof}

\begin{theorem}
	\label{thm:rmlms}
	An itemset $S$ (not containing $x$) is frequent* in $T$ if and only if $S$
	frequent* in the residual tree $T_{R_x}$ of $x$, i.e.,
	\begin{center}
		$S$ is frequent* in $T \Leftrightarrow S$ is frequent* in $T_{R_x}$
	\end{center}
\end{theorem}

\begin{proof} 
	Suppose $S$ is a $k$-itemset. \\
	$S$ is \emph{frequent}* in $T_{R_x}$ \\
	$\Leftrightarrow$ $supp(S,T_{R_x}) \geq \sigma_{k,p}$ (since $\rho_{T_{R_x}}
	= p$) \\
	$\Leftrightarrow$ $supp(S,T) \geq \sigma_{k,p}$ (using
	Observation~\ref{obs:rfreq}) \\
	$\Leftrightarrow$ $S$ is \emph{frequent}* in $T$ (since $\rho_{T} = p$).
\end{proof}

The algorithm IFP\_MLMS merges the following \emph{frequent}* itemsets:
(i)~\emph{frequent}* itemsets obtained by including $x$ with those returned from
the projected tree (shown to be correct by Theorem~\ref{thm:pmlms}),
(ii)~\emph{frequent}* itemsets obtained from the residual tree (shown to be
correct by Theorem~\ref{thm:rmlms}) and (iii)~\emph{1-itemset} $\{x\}$ if it is
\emph{frequent}* in $T$.  The root of the tree $T$ that represents the entire
database has a null prefix-set, and therefore, zero prefix-length.  Hence, all
\emph{frequent}* itemsets mined from that tree are the \emph{frequent} itemsets
in the original transaction database.

\section{Experimental Results}
\label{sec:experiments}

In this section, we report the experimental results of running our algorithms on
different datasets.  We first report the performance of \emph{IFP\_min}
algorithm in comparison with the \emph{Apriori\_min} and \emph{MINIT}
algorithms, followed by that of the \emph{IFP\_MLMS} algorithm.  The benchmark
datasets have been taken from Frequent Itemset Mining Implementations (FIMI)
repository \url{http://fimi.ua.ac.be/data/}.  All experiments were run on a
machine with Dual Core Intel Processor running at 2.4GHz with 8GB of RAM.

\subsection{IFP\_min}

The Accident dataset is characteristic of \emph{dense} and \emph{large}
datasets. Figure~\ref{fig:accident} shows the performance of different
algorithms on the dataset. The IFP\_min algorithm outperforms the MINIT and
Apriori\_min algorithm by exponential factors. Apriori\_min algorithm, due to
its inherent property of performing worse than the pattern growth based
algorithms on dense datasets, performs the worst.

\begin{figure}[t]
\centering
\includegraphics[width=\figwidth]{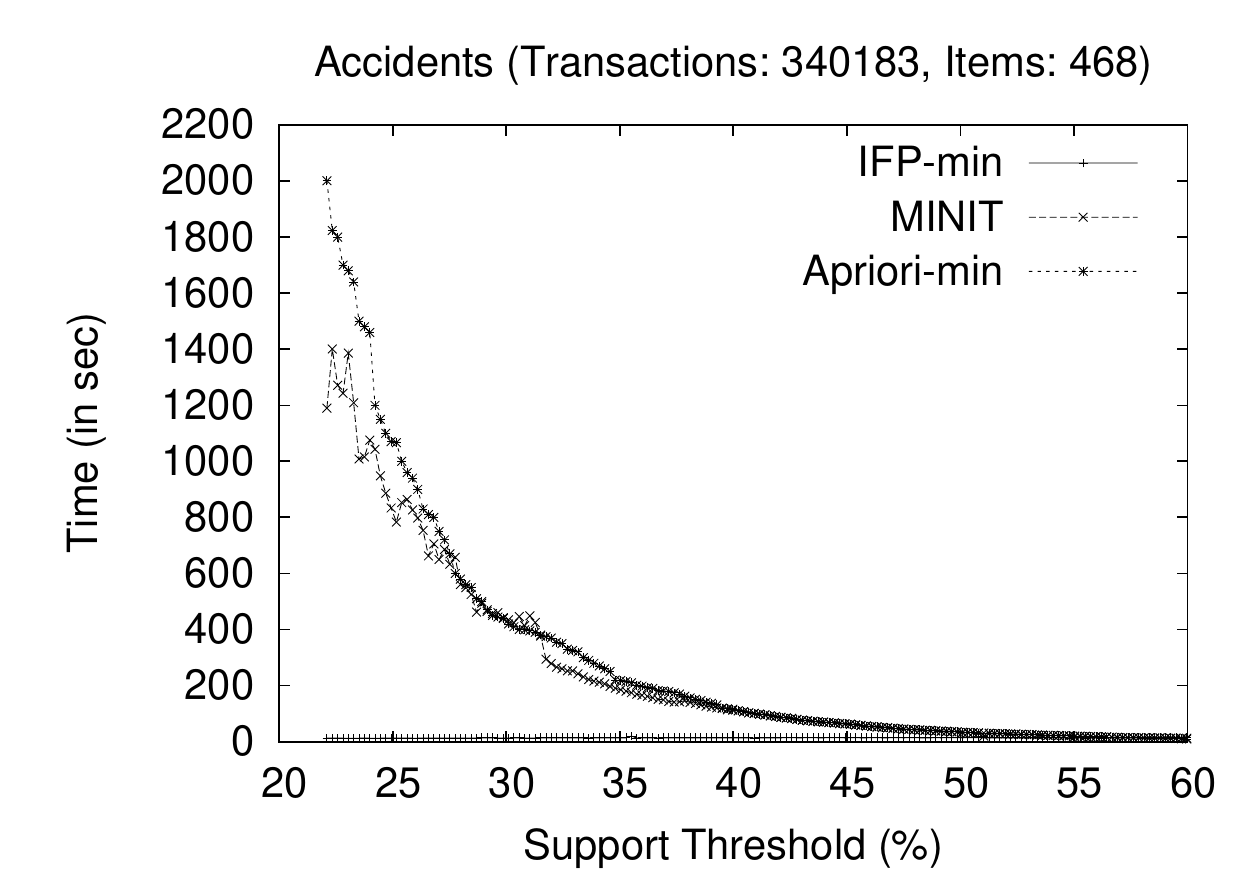}
\caption{Accident Dataset}
\label{fig:accident}
\end{figure}

The Connect (Figure~\ref{fig:connect}), Mushroom (Figure~\ref{fig:mushroom}) and
Chess (Figure~\ref{fig:chess}) datasets are characteristic of \emph{dense} and
\emph{small} datasets. The Apriori\_min algorithm achieves better reduction on
the size of candidate sets. However, when there exist a large number of frequent
itemsets, candidate generation-and-test methods may suffer from generating huge
number of candidates and performing several scans of database for
support-checking, thereby increasing the computational time. The corresponding
computational times have not been shown in the interest of maintaining the scale
of the graph present for IFP\_min and MINIT. 

As can be observed from the figures, dense and small datasets are characterized
by a neutral support threshold below which the MINIT algorithm performs better
than IFP\_min and above which IFP\_min performs better than MINIT.  The MINIT
algorithm prunes an item based on the support threshold and length of the
itemset in which the item is present (\emph{minimum support
property}~\cite{minit}). As the support thresholds are reduced, the pruning
condition becomes activated and leads to reduction in search space. Above the
neutral point, the pruning condition is not effective. In IFP\_min algorithm,
any candidate MII itemset is checked for set membership in a residual database
whereas in MINIT the candidates are validated by computing the support from the
whole database. Due to reduced validation space, IFP\_min outperforms MINIT.

\begin{figure}[t]
\centering
\includegraphics[width=\figwidth]{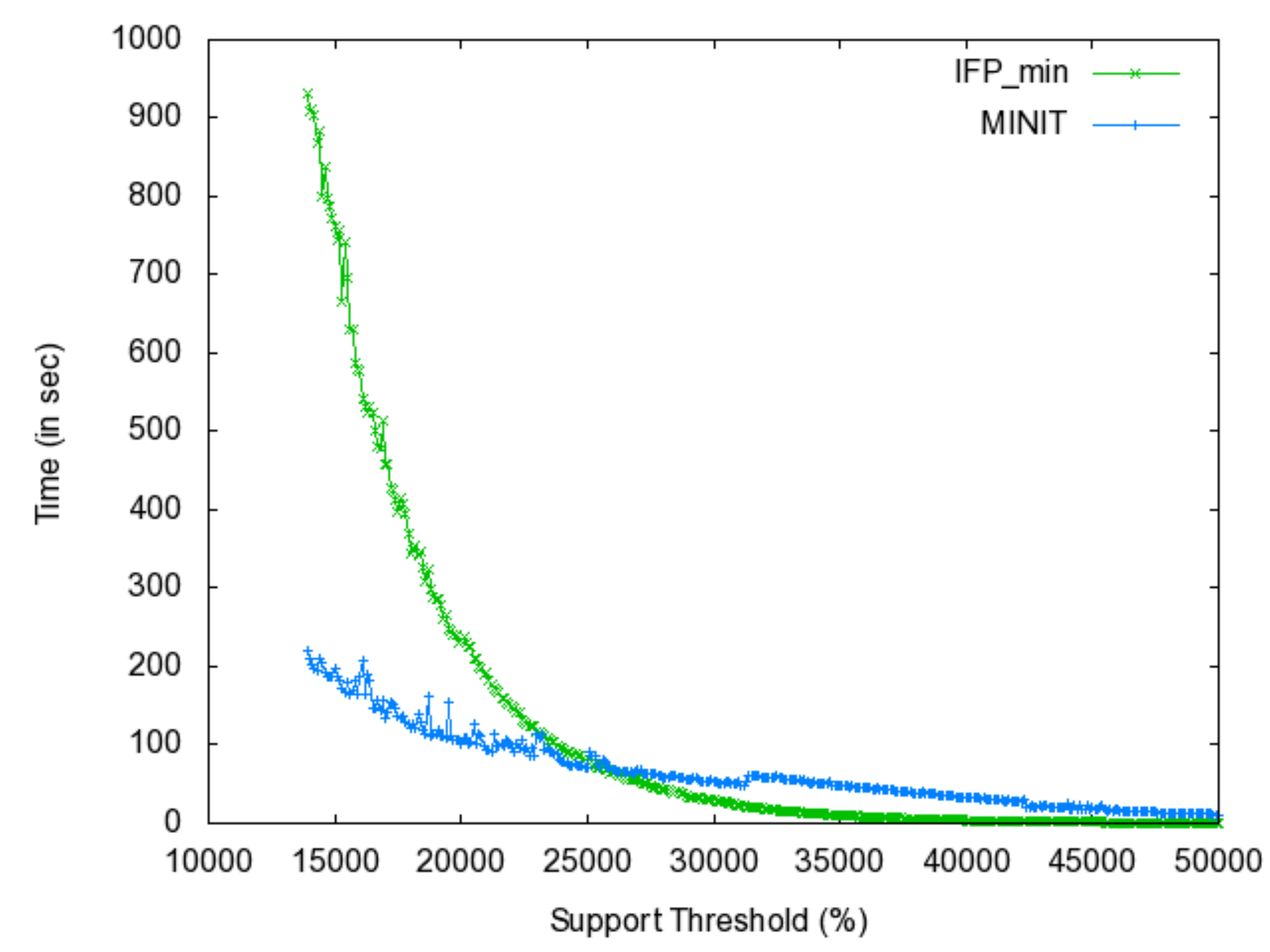}
\caption{Connect Dataset}
\label{fig:connect}
\end{figure}

\begin{figure}[t]
\centering
\includegraphics[width=\figwidth]{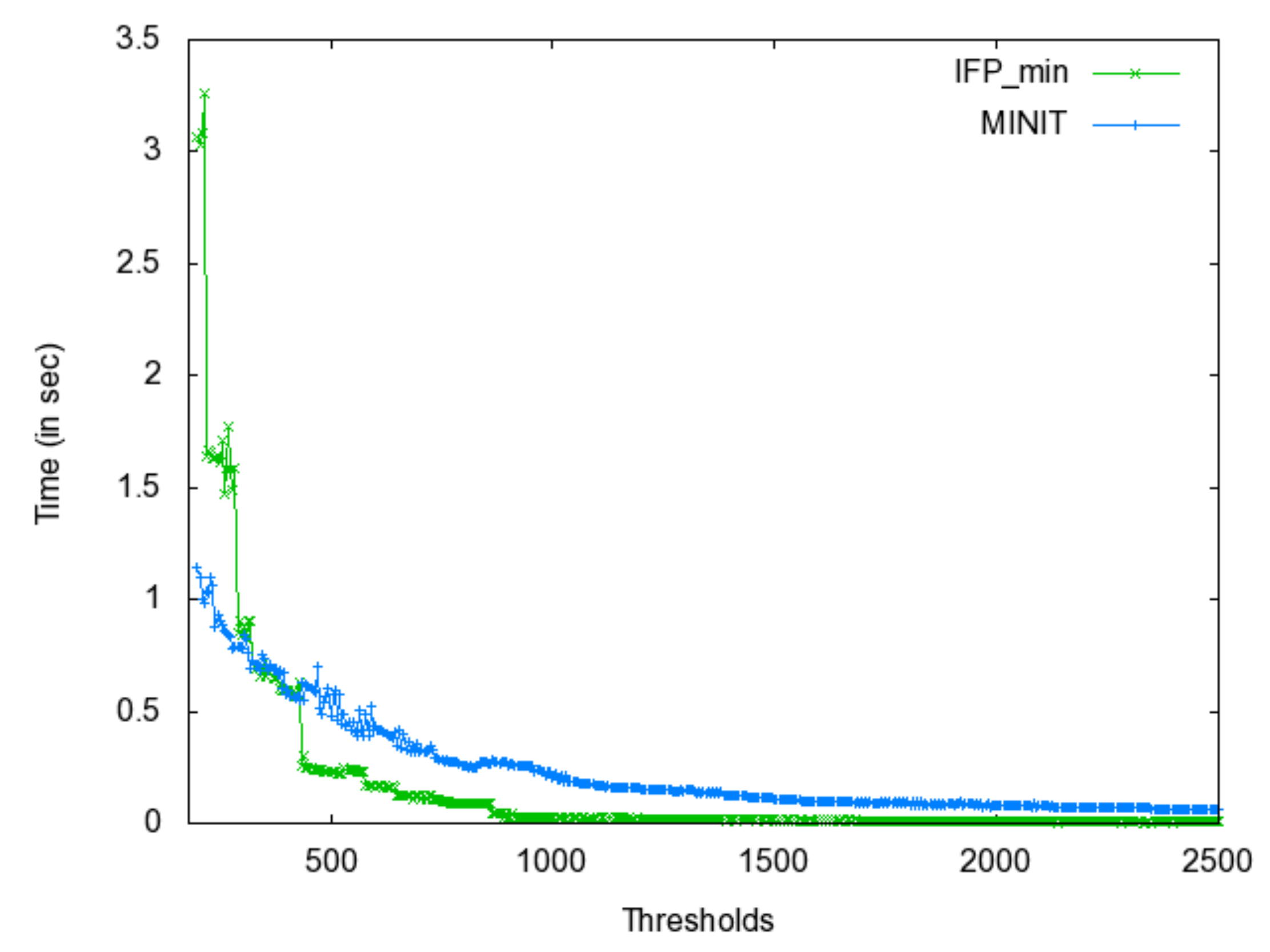}
\caption{Mushroom Dataset}
\label{fig:mushroom}
\end{figure}

\begin{figure}[t]
\centering
\includegraphics[width=\figwidth]{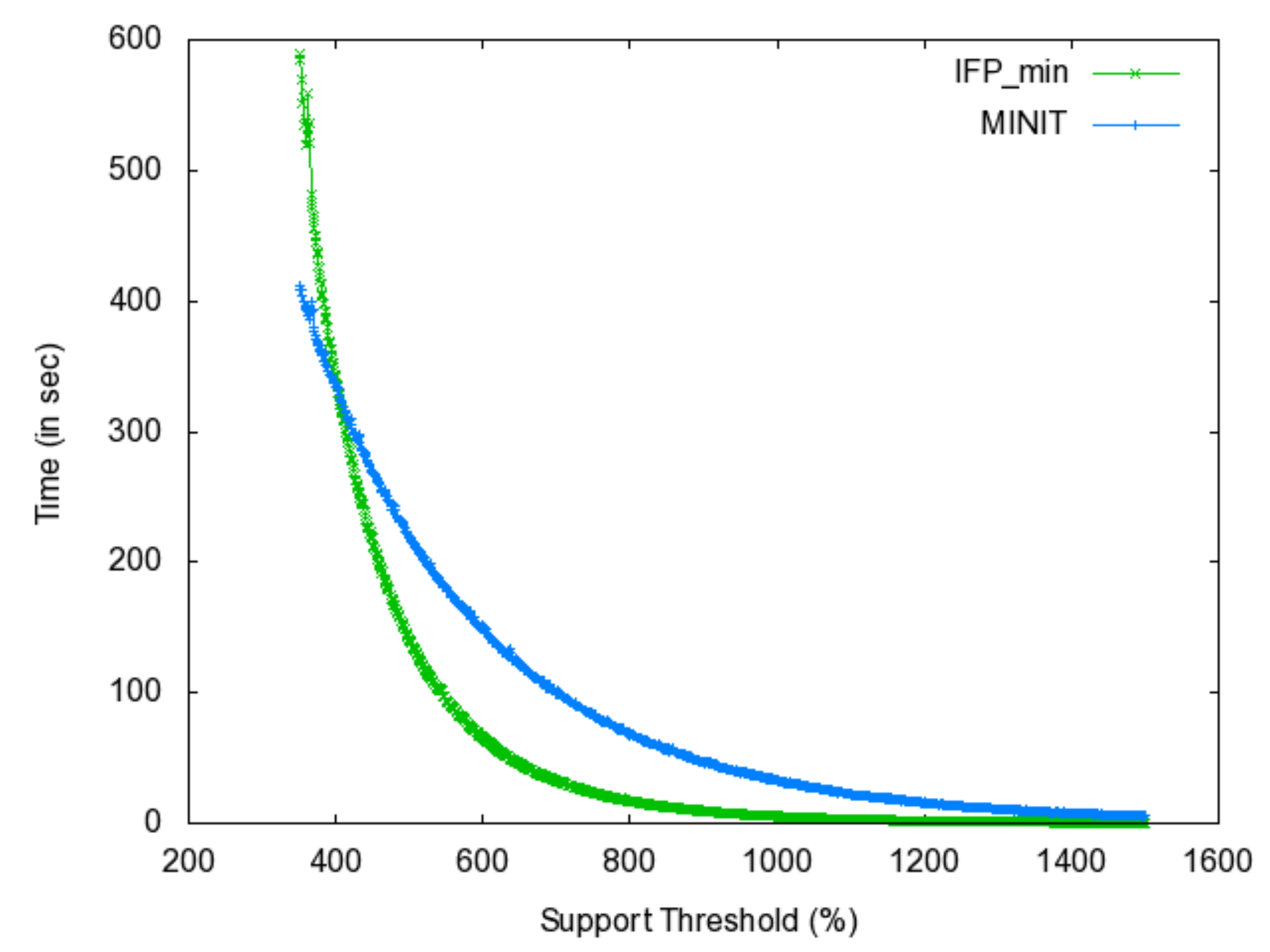}
\caption{Chess Dataset}
\label{fig:chess}
\end{figure}

The T10I4D100K (Figure~\ref{fig:t10}) and T40I10D100K (Figure~\ref{fig:t40}) are
\emph{sparse} datasets. Since Apriori\_min is a candidate-generation-and-test
based algorithm, it halts when all the candidates are infrequent. As such, it
avoids the complete traversal of the database for all possible lengths. However,
both IFP\_min and MINIT, being based on the recursive elimination procedure,
have to complete their full run in order to report the MIIs. This results in
higher computational times for these methods.

\begin{figure}[t]
\centering
\includegraphics[width=\figwidth]{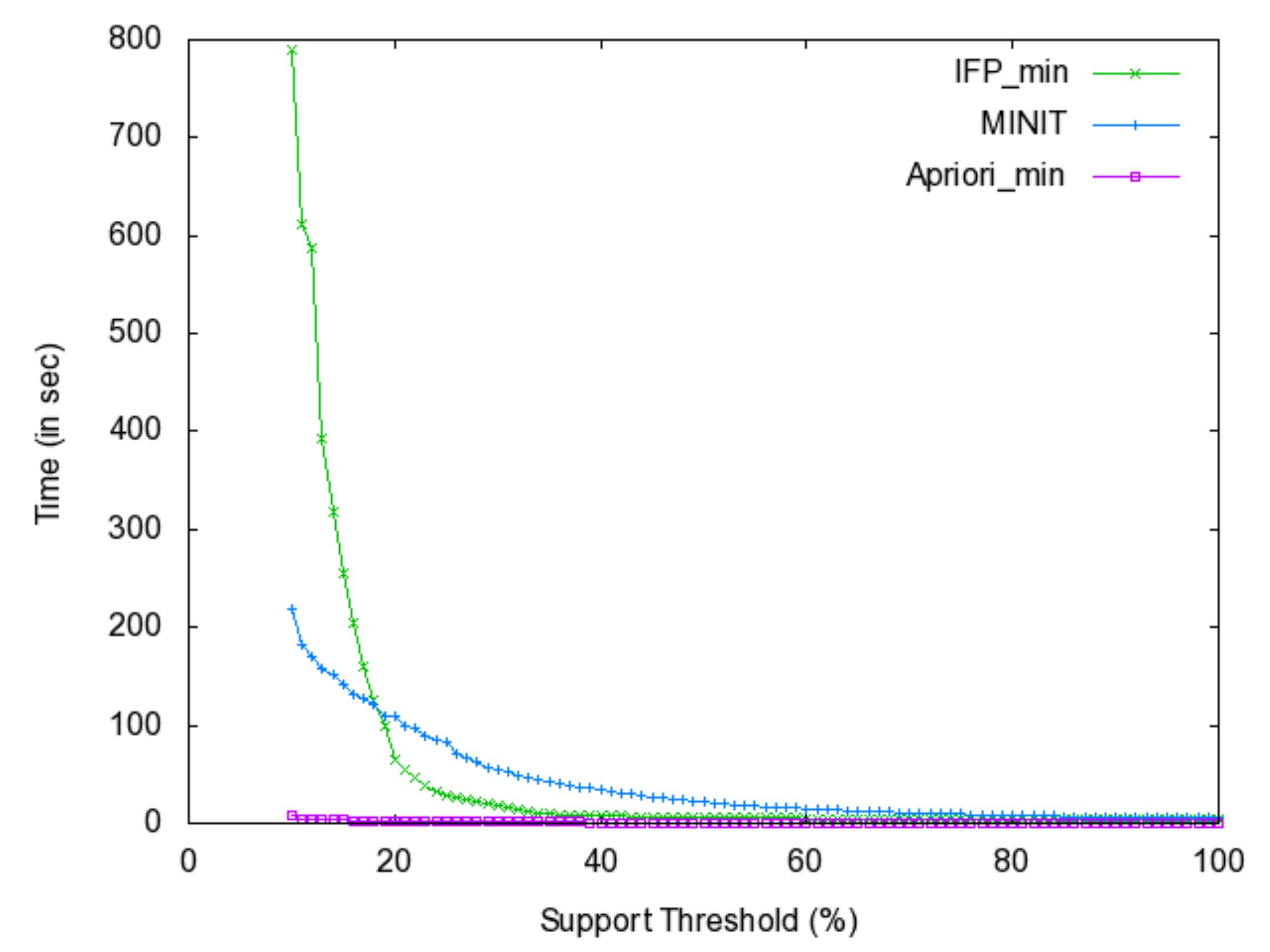}
\caption{T10I4D100K Dataset}
\label{fig:t10}
\end{figure}

\begin{figure}[t]
\centering
\includegraphics[width=\figwidth]{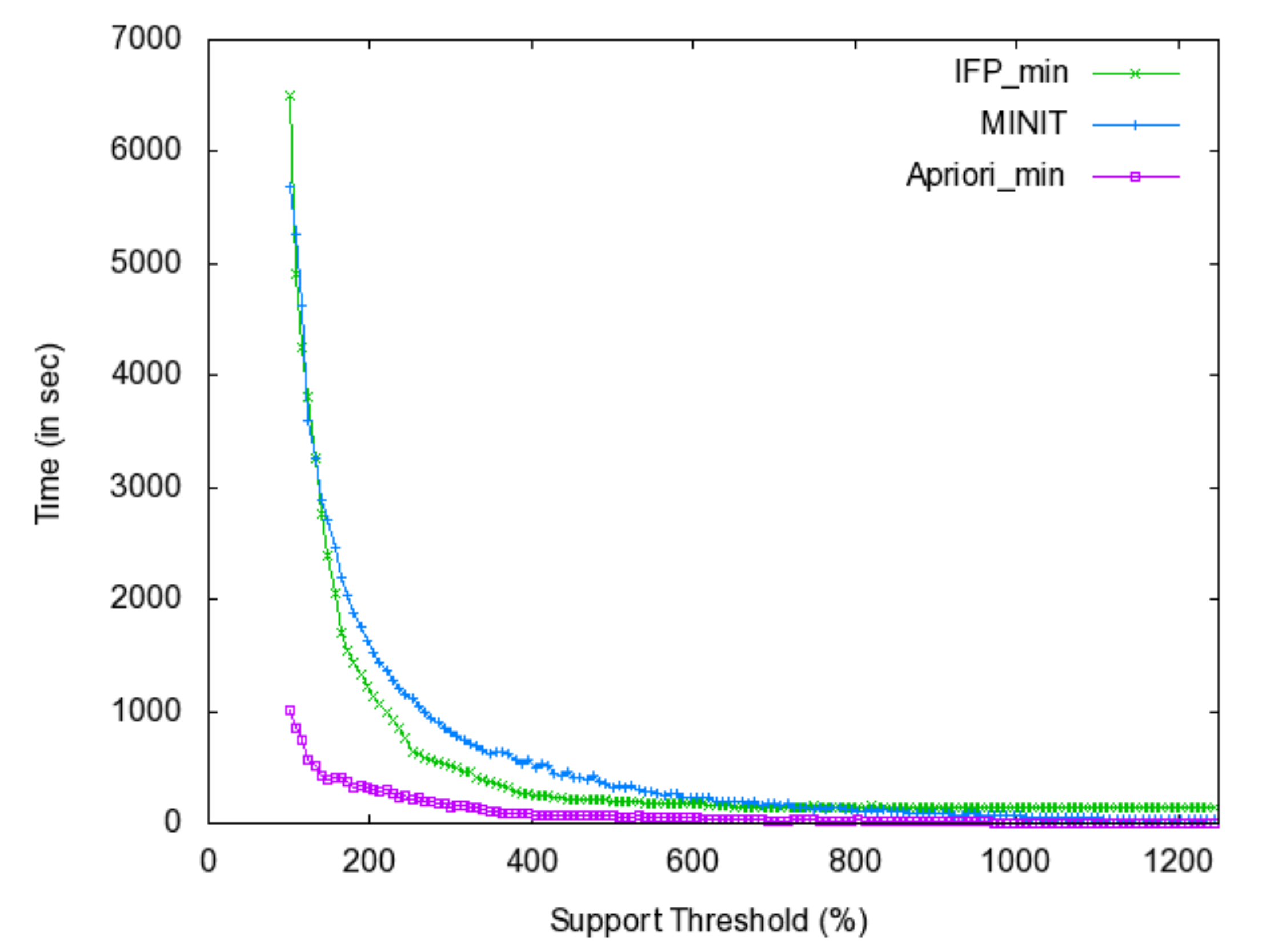}
\caption{T40I10D100K Dataset}
\label{fig:t40}
\end{figure}

On analyzing and comparing the MIIs generated by IFP\_min, Apriori\_min and
MINIT algorithms, we found that the MIIs belonging to Group 2b (i.e., the
itemsets having zero support threshold in the transaction database) are not
reported by the MINIT algorithm, thereby leading to its incompleteness. Based on
the experimental analysis, it is observed that for large dense datasets, it is
preferable to use IFP\_min algorithm. For small dense datasets, MINIT should be
used at low support thresholds and IFP\_min should be used at larger thresholds.
For sparse datasets, Apriori\_min should be used for reporting MIIs. 

\subsection{IFP\_MLMS}

In this section, we report the performance of IFP\_MLMS algorithm in comparison
with the Apriori\_MLMS~\cite{mlms} algorithm. Several real and synthetic
datasets (obtained from \url{http://archive.ics.uci.edu/ml/datasets}) have been
used for testing the performance of the algorithms. For dense datasets, due to
the presence of large number of transactions and items, the Apriori\_MLMS
algorithm crashes for lack of memory space.

\begin{figure}[t]
\centering
\includegraphics[width=\figwidth]{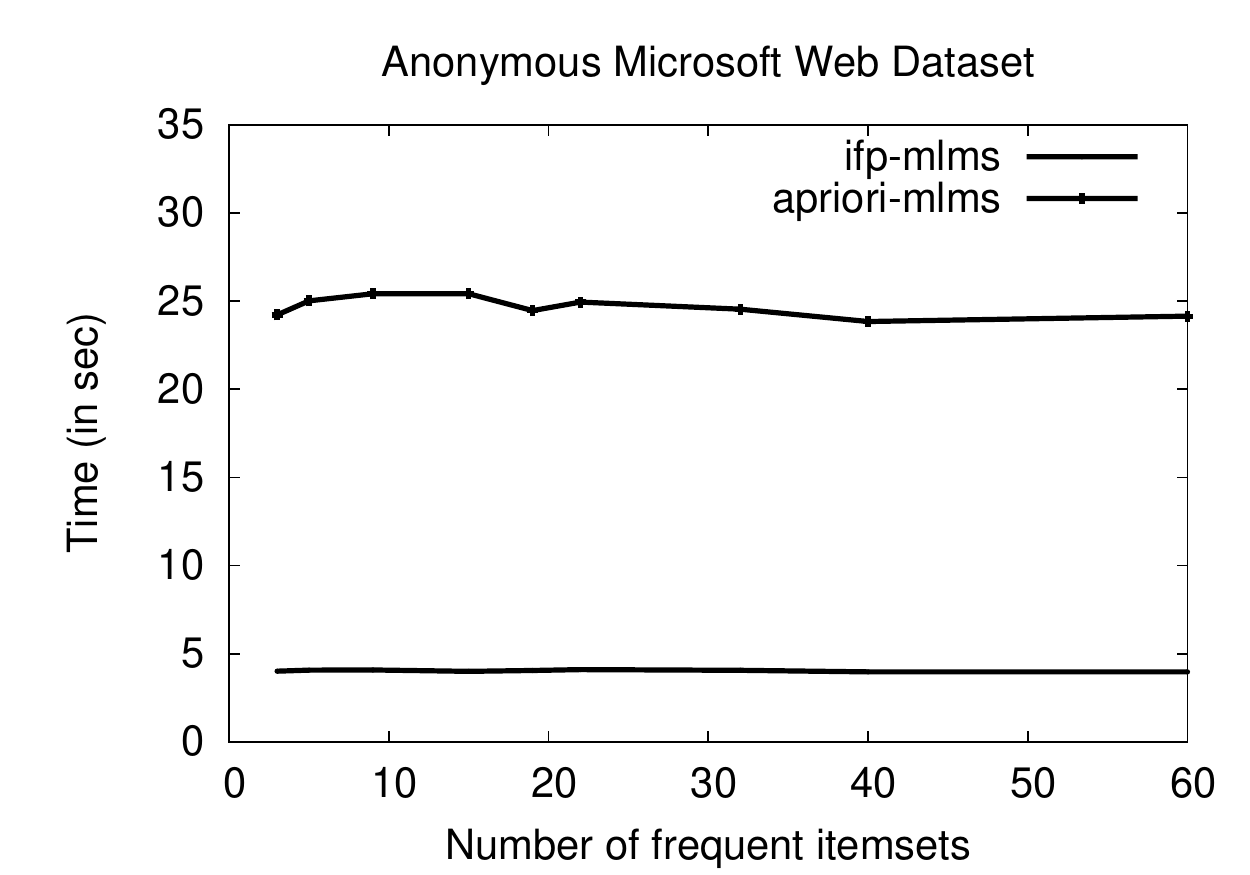}
\caption{The IFP\_MLMS and Apriori\_MLMS algorithms on the Anonymous Microsoft Web Dataset.}
\label{microsoft1}
\end{figure}

\begin{figure}[t]
\centering
\includegraphics[width=\figwidth]{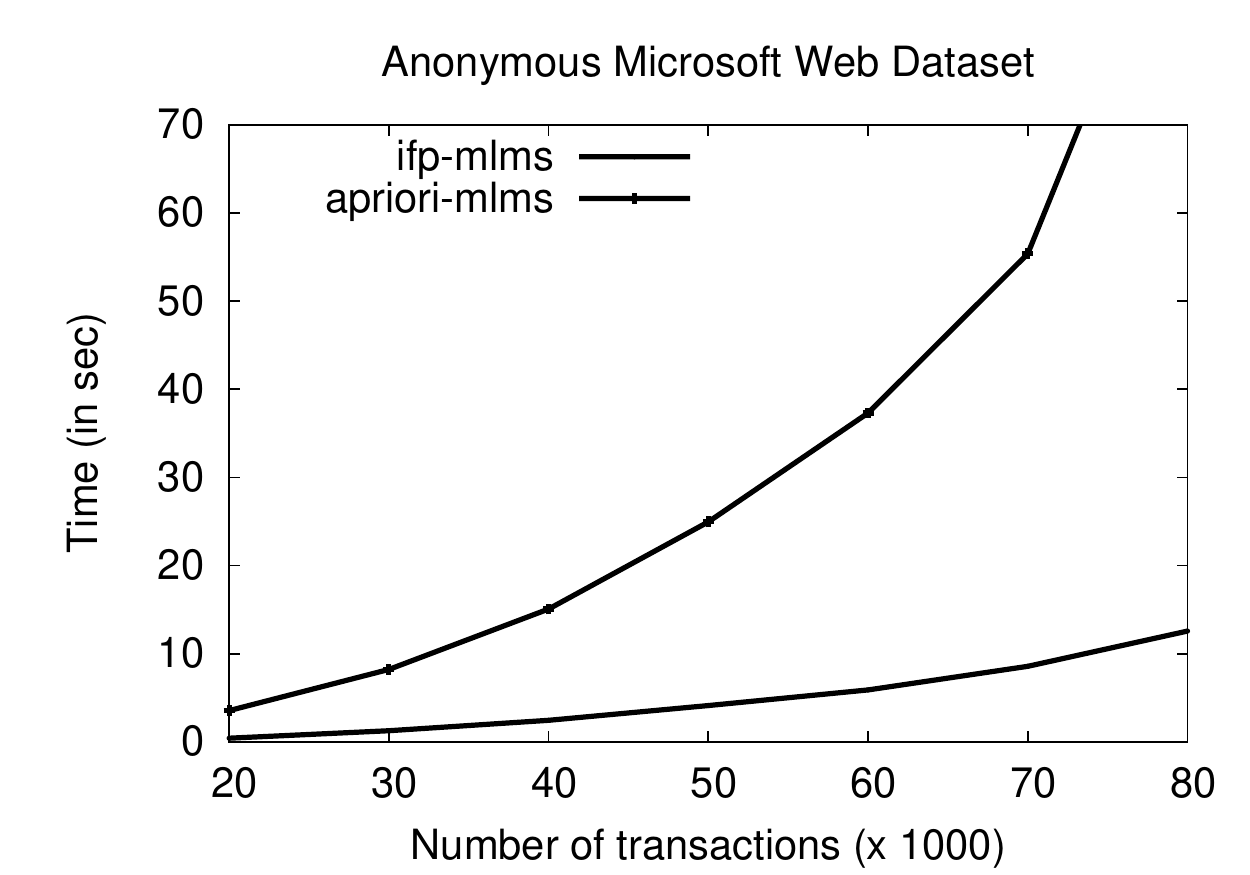}
\caption{The IFP\_MLMS and Apriori\_MLMS algorithms on the Anonymous Microsoft Web Dataset.}
\label{microsoft2}
\end{figure}

The first dataset we used is the Anonymous Microsoft Web dataset that records
areas of \url{www.microsoft.com} each user visited in a one week time frame in
February 1998
(\url{http://archive.ics.uci.edu/ml/datasets/Anonymous+Microsoft+Web+Data}).
The dataset consists of 1,31,666 transactions and 294 attributes.

Figure~\ref{microsoft1} plots the performance analysis of IFP\_MLMS and
Apriori\_MLMS algorithms.  The minimum support thresholds for itemsets of
different lengths were varied over a distribution window from 2\% to 20\% at
regular intervals.  The graph clearly shows the superiority of IFP\_MLMS as
compared to Apriori\_MLMS. We also observe that the time taken by the algorithm
to compute the frequent itemsets is roughly independent of the minimum support
thresholds for different lengths.

\begin{figure}[t]
\centering
\includegraphics[width=\figwidth]{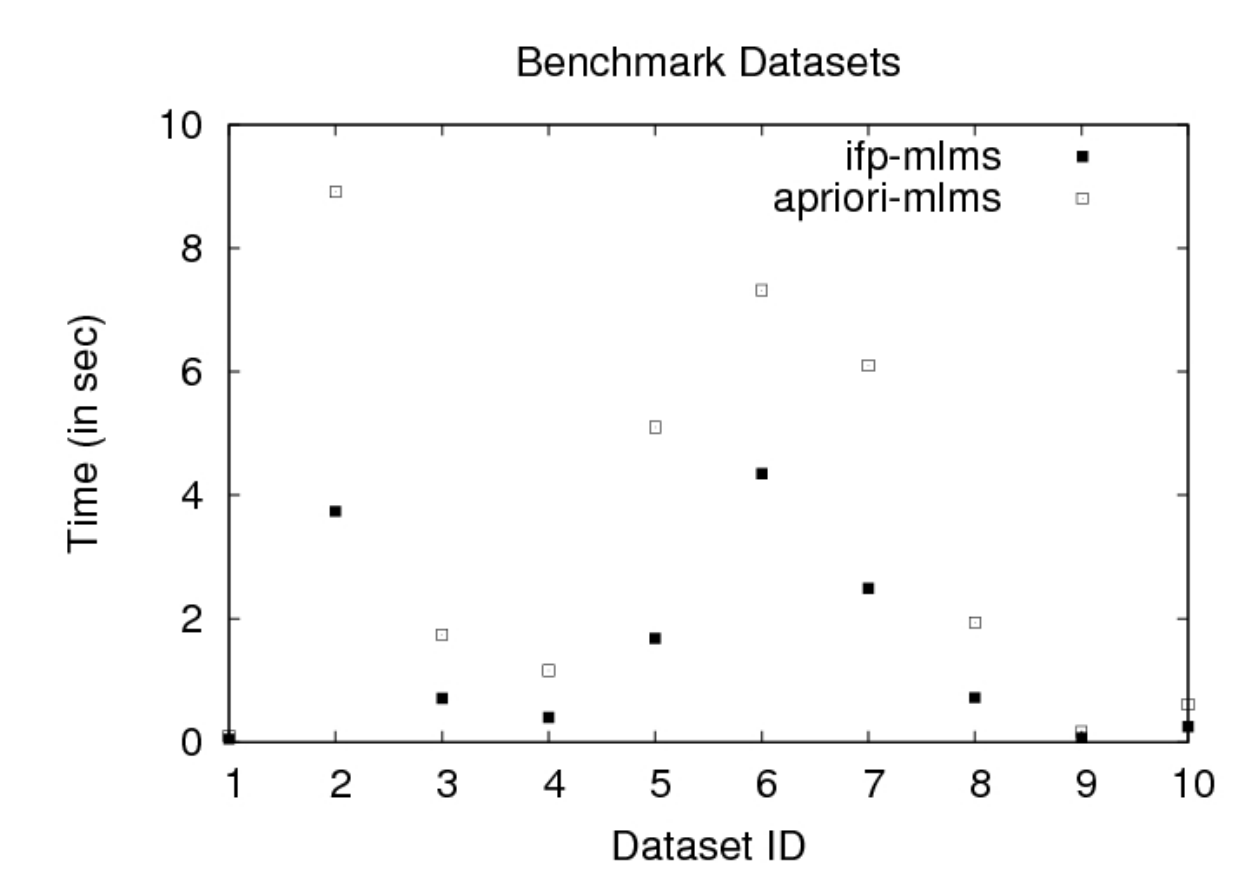}
\caption{The IFP\_MLMS and Apriori\_MLMS algorithms on the datasets given in
Table~\ref{tab:dataset}.}
\label{bench10}
\end{figure}

For the plot shown in Figure~\ref{microsoft2}, the number of transactions are
varied in the Anonymous Microsoft Web dataset.  We observe that the running time
for both the algorithms increases with the number of transactions when the
support thresholds are kept between 3\% to 10\%. However, the rate of increase
in time for Apriori\_MLMS algorithm is much higher and at 80,000 transactions,
Apriori\_MLMS crashes due to lack of memory.

Since the Apriori\_MLMS fails to give results for large datasets, smaller
datasets were obtained from \url{http://archive.ics.uci.edu/ml/datasets/} for
comparison purposes with the IFP\_MLMS. The characteristics of these datasets
are shown in Table~\ref{tab:freqmlms}. For each such dataset, the corresponding
time for IFP\_MLMS and Apriori\_MLMS are plotted in Figure~\ref{bench10}. The
support threshold percentages are kept the same for all the datasets and were
varied between 10\% and 60\% for the different length itemsets.

\begin{table}[t]
\centering
\begin{tabular}{|c|c|c|c|}
\hline
\multirow{2}{*}{ID} & \multirow{2}{*}{Dataset} & Number of & Number of \\
& & items & transactions \\
\hline
\hline
1 & machine & 467 & 209\\
2 & vowel\_context & 4188 & 990\\
3 & abalone & 3000 & 4177\\
4 & audiology & 311 & 202\\
5 & anneal & 191 & 798\\
6 & cloud & 2057 & 11539\\
7 & housing & 2922 & 506\\
8 & forest\_fires & 1039 & 518\\
9 & nursery & 28 & 12961\\
10 & yeast & 1568 & 1484\\
\hline
\end{tabular}
\caption{Details of smaller datasets.}
\label{tab:dataset}
\end{table}

The above results clearly show that the IFP\_MLMS algorithm outperforms
Apriori\_MLMS. During experimentation, we found the running time of both
IFP\_MLMS and Apriori\_MLMS algorithms to be independent of support thresholds
for the MLMS model. This behavior is attributed to the absence of downward
closure property~\cite{ragrawal} of frequent itemsets in the MLMS model, unlike
that of the single threshold model.

Further, consider an alternative FP-Growth algorithm for the MLMS model that
mines all frequent itemsets corresponding the lowest support threshold
$\sigma_{low}$ and then filters the $\sigma_k$ frequent $k$-itemsets to report
the frequent itemsets.  In this case, a very large set of frequent itemsets is
generated that renders the filtering process computationally expensive. The
pruning based on $\sigma_{low}$ in IFP\_MLMS ensures that the search space is
same for both the algorithms. Moreover, the filtering required for the former is
implicitly performed in IFP\_MLMS, thus making IFP\_MLMS more efficient.

\section{Conclusions}
\label{sec:conc}

In this paper, we have introduced a novel algorithm, IFP\_min, for mining
minimally infrequent itemsets (MIIs).  To the best of our knowledge, this is the
first paper that addresses this problem using the pattern-growth paradigm.  We
have also proposed an improvement of the Apriori algorithm to find the MIIs.
The existing algorithms are evaluated on dense as well as sparse datasets.
Experimental results show that: (i)~for large dense datasets, it is preferable
to use IFP\_min algorithm, (ii)~for small dense datasets, MINIT should be used
at low support thresholds and IFP\_min should be used at larger thresholds and
(iii)~for sparse datasets, Apriori\_min should be used for reporting the MIIs.

We have also designed an extension of the algorithm for finding frequent
itemsets in the multiple level minimum support (MLMS) model.  Experimental
results show that this algorithm, IFP\_MLMS, outperforms the existing
candidate-generation-and-test based Apriori\_MLMS algorithm.

In future, we plan to utilize the scalable properties of our algorithm to mine
maximally frequent itemsets.  It will be also useful to do a performance
analysis of IFP-tree in parallel architecture as well as extend the IFP\_min
algorithm across different models for itemset mining, including interestingness
measures.

\bibliographystyle{abbrv}
\balance
\bibliography{comad}

\end{document}